%% file: main.tex
\documentclass[a4paper,11pt]{article}
\usepackage[T1]{fontenc}
\usepackage{natbib}
\usepackage[margin=0.7in, a4paper, total={170mm,257mm}, left=20mm, top=20mm]{geometry} % [margin=1in]
\usepackage{authblk}

\usepackage{graphicx}
\usepackage{url}
\usepackage[pdftex,dvipsnames]{xcolor}
\usepackage{xspace}
\usepackage{tablefootnote}
\usepackage{amsmath, amsfonts, amsthm}
\usepackage[ruled,vlined]{algorithm2e}
\usepackage{hyperref}
\usepackage[capitalize]{cleveref}
\usepackage{csquotes}

\usepackage{tikz}
\usepackage{pgfplots}
\usepgfplotslibrary{groupplots}

\usepackage{booktabs}

\usepackage{subcaption}
\newtheorem{theorem}{Theorem}[section]

\newtheorem{lemma}[theorem]{Lemma}
\newtheorem{definition}{Definition}
\newtheorem{remark}{Remark}

\newcommand{\R}{\mathbb R} % Real numbers

\newcommand{\datasets}{\mathcal{D}}
\usepackage{caption}
\usepackage{subcaption}
\usepackage{pifont}
\usepackage{blindtext}
\usepackage{graphicx}               % Include graphics
\usepackage{appendix}               % Use an appendix
\usepackage{textpos}                % Easy text positioning
\usepackage{typearea}               % Standard for defining geometry
\usepackage{tikz}                   % TikZ pictures
\pgfplotsset{compat=1.18}
\usepackage{amsmath,amssymb,mathtools} % Math
\usepackage{scrhack}
\usepackage{listings}               % Listings
\usepackage{paralist}
\usepackage{stmaryrd}               % Lightning arrow \lightning
\usepackage{textcomp}
\usepackage{hyperref}

\begin{document}
    \include{res/defs}
    \date{}
    
    \title{\Large{Understanding the Theoretical Guarantees of DPM}}
    
    \newcommand{\repthanks}[1]{\textsuperscript{\ref{#1}}}
    \makeatletter
    \patchcmd{\maketitle}
      {\def\thanks}
      {\let\repthanks\repthanksunskip\def\thanks}
      {}{}
    \patchcmd{\@maketitle}
      {\def\thanks}
      {\let\repthanks\@gobble\def\thanks}
      {}{}
    \newcommand\repthanksunskip[1]{\unskip{}}
    \makeatother
    \author[1]{Yara Sch\"utt}
    \author[1]{Esfandiar Mohammadi}
    
    \affil[1]{University of L\"ubeck, L\"ubeck, Germany}
    
    \maketitle
    
    \begin{abstract}
        In this study, we conducted an in-depth examination of the utility analysis of the differentially private mechanism (DPM). The authors of DPM have already established the probability of a good split being selected and of DPM halting. In this study, we expanded the analysis of the stopping criterion and provided an interpretation of these guarantees in the context of realistic input distributions. Our findings revealed constraints on the minimum cluster size and the metric weight for the scoring function. Furthermore, we introduced an interpretation of the utility of DPM through the lens of the clustering metric, the silhouette score. Our findings indicate that even when an optimal DPM-based split is employed, the silhouette score of the resulting clustering may still decline. This observation calls into question the suitability of the silhouette score as a clustering metric. Finally, we examined the potential of the underlying concept of DPM by linking it to a more theoretical view, that of $(\xi, \rho)$-separability. This extensive analysis of the theoretical guarantees of DPM allows a better understanding of its behaviour for arbitrary inputs. From these guarantees, we can analyse the impact of different hyperparameters and different input data sets, thereby promoting the application of DPM in practice for unknown settings and data sets.
    \end{abstract}
    \newpage
    \tableofcontents
    \include{main/introduction}
    \include{main/stoppingCriterion}
    \include{main/silhScore}
    \include{main/xiRho}
    \include{main/conclusions}
    
    % The guidelines say, you should use 'cell'. Using cell in latex is quite
    % buggy, so you might want to use 'apalike' instead. Cell is based on apalike
    % and very similar. You may also try out amsalpha, alpha and plain.

    \bibliographystyle{plain}
    \bibliography{ref}

\end{document}

%% file: res/defs.tex
\newcommand{\dist}{\text{d}}
\newcommand{\emin}{e_{\text{min}}}
\newcommand{\outerQuantile}{Q_O}
\newcommand{\innerQuantile}{Q_I}

%% file: main/introduction.tex
\section{Introduction}
\label{cha:introduction}
% Why do we need utility bounds for algorithms
There is recent work that introduces a clustering algorithm that is based on a robust strategy finding smart separators: DPM \cite{DPM}. We show utility proofs for DPM.
Theoretical bounds on the utility characterise the conditions under which an algorithmic result is useful. For sufficiently tight utility bounds, dependencies of hyperparameters and assumptions about the input data can be derived. Thus, when choosing appropriate hyperparameters and understanding the behaviour of an algorithm for different settings, such utility bounds can be of use. 
A substantial number of research papers employ empirical evaluations and ablation studies to elucidate the utility of algorithms. However, the experimental outcomes are constrained to the data sets that are evaluated and the evaluated combinations of values for the hyperparameters.

Experimental evaluations can test some dependencies between hyperparameters and properties of the input, while sufficiently tight theoretical bounds precisely characterise these relations. However, relying on the experimental evaluation alone may result in the omission of potential relations between hyperparameters and inputs, leading to an incomplete characterisation of an algorithm's behaviour regarding the utility.

In this study, we present theoretical bounds for a DP clustering algorithm, DPM, which characterise its utility in diverse settings without assumptions on the input data. By adding assumptions on the settings and the input data, such as that there is always at least one central split or that the data points are drawn from a multiple Gaussian distributions, the utility bounds become more precise.\\

% Utility-privacy: DP
The utility-privacy-tradeoff inherent to differentially private mechanisms is a fundamental aspect of their utility. A high level of privacy inevitably introduces a considerable amount of noise, which in turn reduces the potential for utility. Similarly, a high level of utility necessitates a lower level of noise, thereby only allowing for a smaller degree of privacy to be guaranteed. In Differential Privacy, the privacy parameters, denoted by the symbols $\varepsilon$ and $\delta$, are used to quantify the privacy guarantees of a mechanism. For instance, in the case of the differentially private mechanism proposed in \cite{DPM}, the added noise for a fixed privacy budget depends on both the input data set and some hyperparameters. The bounds on the utility of a mechanism depend on the hyperparameters of the mechanism, which also include the privacy parameters. \\

Firstly, we introduce the concept of Differential Privacy, which is employed by state-of-the-art privacy-preserving algorithms. As we analyse the utility guarantees of DPM, we utilise the identical definitions of privacy as in \cite{DPM}. Prior to presenting the clustering algorithm DPM and the remaining question regarding theoretical analysis, we provide some background on the field of clustering in general. Finally, we discuss pertinent related work.

\subsection{Differential Privacy}
The privacy notion Differential Privacy (DP) is based on the assumption that the input data sets consist of data contributed by different individuals, some of which may be sensitive. Consequently, each individual can potentially influence the behaviour of a mechanism and its output. Despite the necessity of the mechanism to utilise the data from the individuals for optimal performance, the privacy of any individual must be preserved. This necessitates that the output should be robust to small changes (one individual's data) in the input data set.
Consequently, in order to analyse the DP guarantees of a mechanism, it is necessary to determine the largest impact that one individual can have on the output. 
To determine this impact, we consider neighbouring data sets, which are data sets that only differ by the data of one individual. 
A data set consists of $n$ $d$-dimensional data points (individuals) $D = \{x_0, \dots, x_{n-1}\} \in \mathbb{R}^{d \times n}$. The conjunction of all possible data sets is the set of all data sets $\datasets = \{D \in \mathbb{R}^{d \times n}|~ d,n \in \mathbb{N}_+\}$.
% Neighbouring data sets
\begin{definition}[Neighbouring data sets]
Given two data sets $D_0,D_1 \in \datasets$, we say $D_0$ and $D_1$ are neighbouring if $D_1= D_0 \cup \{x\}$ or $D_0= D_1 \cup \{x\}$ for a data point $x\in \mathbb{R}^d$. We will denote neighbouring data sets as $D_0 \sim D_1$.
\end{definition}

%% Sensitivity
In Differential Privacy, we consider the maximum leakage for a data point, i.e. for all neighbouring data sets, we  consider the pair that results in the largest difference in the output of a mechanism. This difference is called the sensitivity of a mechanism.
\begin{definition}[Sensitivity (Def. 2.2\cite{DPM})]
\label{def:sensitivity}
    Given two neighbouring data sets $D_0, D_1$ and some set $X$. A function $f:\datasets \times X \rightarrow \mathbb{R}$ has sensitivity $\Delta_f$ iff.\ $\Delta_f \ge \max_{D_0 \sim D_1} |f(D_0) - f(D_1)|$. 
    A function $f$ with sensitivity $\Delta_f \in \mathbb{R}$ is a $\Delta_f$-bounded query.
\end{definition}

% DP
For all neighbouring data sets, we have to consider the pair that results in the largest difference in the output of a mechanism. The bound of this difference directly gives us the privacy parameters $\varepsilon$ and $\delta$.
\begin{definition}[$(\varepsilon, \delta)$-DP]
    A mechanism $M$ with $M : \datasets \rightarrow A$ preserves \emph{$(\varepsilon,\delta)$-Differential Privacy (short: $(\varepsilon,\delta)$-DP)} for some $\varepsilon \in \mathbb{R}_{>0}$ and $0 \le \delta \leq 1$ if for all neighbouring data sets $D_0,D_1\in \datasets$ with $D_0 \sim D_1$ and all possible observations $O\subseteq A$:
    \[\textstyle\Pr[{M(D_0) \in O}] \le \exp(\varepsilon) \textstyle\Pr[{M(D_1) \in O}] + \delta \text{.}\]
\end{definition} 

% Composition theorems
Algorithms that preserve DP are often compositions of multiple mechanisms that each preserve DP. Sequentially composing $\ell$ mechanisms that each preserve $(\varepsilon, \delta)$-DP results in an algorithm that preserves $(\ell\varepsilon, \ell\delta)$-DP. Composing mechanisms that preserve $(\varepsilon, \delta)$-DP and each take disjoint subsets of the input data set as input result in an algorithm that preserves $(\varepsilon, \delta)$-DP.

The noisy argmax algorithm, the Exponential Mechanism is an important building block of DPM. The bounds on the utility guarantees of the Exponential Mechanism, determine the behaviour of DPM and thus affect the utility guarantees.
\begin{definition}[Exponential Mechanism (Def. 2.5 \cite{DPM})]
\label{def:expMech}
Given $S \in \datasets$ and $\varepsilon >0$, a $\Delta_{f}$-bounded function $f$: $\Delta_f \ge \max_{S\sim S', s\in W_S \cup W_{S'}} |f(S,s) - f(S',s)|$, where  $W_S:= dom(f(S,\cdot))$ is the domain of the second input of $f$. Then the Exponential Mechanism $M_{E}$ takes $S, f, \varepsilon$ as input and draws an element $s \in W$ with probability
\[\mathrm{pmf}_{M_{E}(S,f,\varepsilon)}(s)=\frac{\exp(\varepsilon f(S,s)/\Delta_{f})}{\sum_{s'\in W_S}\exp(\varepsilon f(S,s')/\Delta_{f})}\text{.}\]
\end{definition}

\subsection{Clustering}
The objective of clustering algorithms is to identify groups within data sets, whereby the data points within the same group are similar to one another but dissimilar to those in other groups. Such groups are also referred to as clusters. Each cluster is associated with a representative, also referred to as the cluster centre. The objective of identifying clusters is an unsupervised process, whereby the actual clusters are not known a priori. Consequently, the quality of a clustering result cannot be evaluated based on its accuracy. In the context of evaluating a clustering algorithm through experimentation, the use of accuracy as a metric remains a viable approach, particularly when the input data set is labelled. An additional metric for evaluating a clustering result is the inertia. The inertia of a clustering result is defined as the sum of the distances between each data point and its respective cluster centre. One limitation of this metric is that the lowest possible inertia can always be achieved if there are $n$ cluster centres and each data point is assigned to its own. This suggests that an increase in the number of cluster centres may enhance the inertia, even if it does not necessarily lead to an improvement in the clustering result. The metric silhouette score assesses the ratio of the intra-distance (distance to the own cluster centre) to the inter-distance (distance to the next cluster centre) for each data point. As with the inertia, the silhouette score tends to improve with an increase in the number of cluster centres. Further insights into the limitations of the silhouette score can be found in \cref{cha:silhScore}.

\subsection{Previous Work: DPM}
\label{sec:dpm}
% Main idea: tree based approach with level
The differentially private clustering algorithm, DPM, recursively divides an input data set into subsets until only one cluster remains per subset. The objective is to ensure that the clusters found represent the underlying data structure accurately. 
It is therefore crucial to ascertain the optimal location for the separation of the data points. In the initial phase, the candidates for splitting are generated for each dimension by dividing the specified range into intervals of equal size. The midpoint of each interval represents a potential split point. Subsequently, the splits are evaluated in accordance with a pre-defined scoring function, with the objective of identifying those that will yield optimal clustering results. The scoring function comprises two metrics: centreness and emptiness. The emptiness metric gives preference to split candidates with a smaller number of data points in comparison to a greater number of data points. Consequently, the selected split candidates are those that divide dense areas, as opposed to those that divide dense areas and split up clusters. The centreness metric prioritises split candidates that are not in proximity to the border, but rather situated in closer proximity to the centre of the data points. Subsequently, the selected split is implemented, resulting in the division of the current set in accordance with the specified criteria. Subsequently, DPM determines whether the minimum cluster size has been exceeded for one of the subsets. In the event that this is the case, the set that was previously split is added to the set of clusters. Otherwise, if the maximum recursion level has not yet been reached, the procedure is repeated for each subset. For an illustrative example, see \cref{fig:mondrianClustering}.
Subsequently, following the final recursion, a representative is calculated and returned for each cluster, along with the clusters themselves. 
\begin{figure}[ht!]
    \captionsetup{format=plain}
    \centering
    \includegraphics[scale=0.8]{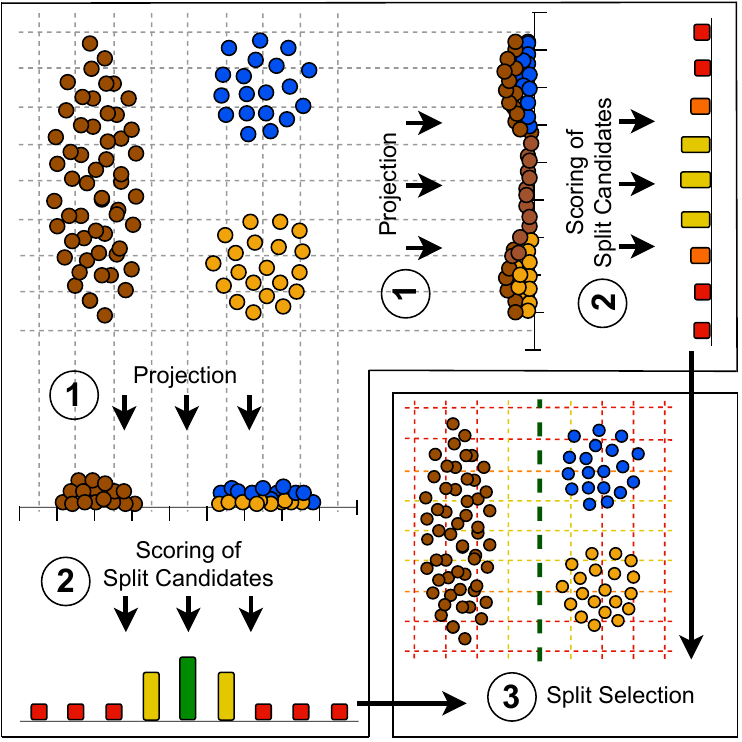}
    \caption{An iteration of the recursive step in the DPM algorithm. \ding{192} The data points are projected onto each dimension, and then multiple split candidates are generated based on a fixed split interval size that is calibrated to the data set. \ding{193} A scoring function, dependent on the specific clustering goal, assigns a score to each split candidate. \ding{194}  The split candidate with the highest score is selected with high probability to subdivide the data set into two disjoint subsets. This procedure is then recursively repeated until only a few elements remain in each subset. \cite[Fig. 2]{DPM}}
    \label{fig:mondrianClustering}
\end{figure}
\\
% Privacy
As DPM is a privacy-preserving clustering algorithm, each step must be designed in a way that ensures the preservation of data privacy, such that no information is leaked. The number of data points in a subset is employed to ascertain whether DPM terminates and to calculate the score of the split candidates. Consequently, DPM solely utilises the noisy number of elements. To select a split candidate that is optimal with respect to a defined scoring function, we can utilise the Exponential Mechanism, which is a noisy argmax. To scale the noise of the Exponential Mechanism, we must ascertain the sensitivity of the scoring function.
In order to determine the representative of a cluster, DPM employs a differentially private averaging algorithm. The authors of \cite{DPM} demonstrate that all subroutines, in addition to their composition, preserve Differential Privacy.\\

% Utility
The authors presented experimental findings pertaining to the efficacy of the method in question, as well as theoretical constraints. The discrepancy between the number of elements and the (shifted) noisy count was quantified. In order to analyse the behaviour of DPM with respect to any given input data, it is necessary to employ the bounds on the utility of the Exponential Mechanisms. Firstly, the probability of a split being selected that is $t'$-central (i.e. a centreness value of at least $t'$) is demonstrated. Subsequently, for a $t'$-central split, the difference between the emptiness of a selected split and the optimal emptiness is bounded once more, utilising the utility bounds of the Exponential Mechanism.

\subsubsection{Scoring Function}
In order to select split candidates that separate dense areas of data points, DPM employs a scoring function that takes into account both the density and the position of a split candidate. The density of a split candidate is quantified by the metric emptiness, which is defined as the number of elements in the split candidate relative to the total number of elements. In order to maximize the score, the difference between the optimal emptiness of 1 and this proportion is computed. The position of a split candidate is defined as the rank of the split candidate if it were to be inserted into the sorted data points along the corresponding dimension.
The rank of a split candidate is employed to guarantee that divisions situated nearer to the centre are prioritised over those situated in close proximity to the border. If the metric emptiness were the sole criterion, splits situated in proximity to the border would be selected, given that they typically encompass a limited number of elements, if any. As a consequence, this would result in the formation of imbalanced splits and the creation of splits that contravene the minimum cluster size. Therefore, the centreness metric serves to reduce the relative score of these split candidates. The precise position of the split candidate is inconsequential; the objective is to avoid splits at the borders. Accordingly, an inner and an outer quantile are defined using the parameter $q$. For each quantile, a linear function is defined with the objective that the slope of the function in the outer quantile is steeper than that of the inner quantile. We ensure that with a second parameter $t$, which is the maximum centreness value of a split candidate in the outer quantile and the minimum centreness value in the inner quantile. 
A split candidate with the median as a rank is assigned the optimal centreness score of $1$. The scoring function, as implemented in DPM, is illustrated in \cref{fig:scoring_func} for an example dataset (orange and blue dots). All data points are projected onto a single axis. For this axis, the metrics emptiness and centre are calculated for some centreness parameter $t,q$ and split interval size $\beta$.
\begin{figure}[ht!]
    \captionsetup{format=plain}
    \centering
    \includegraphics[scale=1.76,trim={0.25cm 0.25cm 0.25cm 0.25cm},clip]{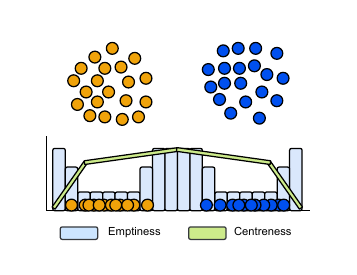}
    \caption{The following visualisation depicts the scoring function employed by DPM for the evaluation of split candidates. In order to assign a score, the data points are projected onto each dimension. Subsequently, for each split candidate, the emptiness (light blue) and centreness (light green) are computed. \cite[Fig. 3]{DPM}}
    \label{fig:scoring_func}
\end{figure}

\begin{definition}[Scoring function]
\label{def:score}
Given a set $S$ and the subscore weight $\alpha >0$ and further the centreness parameters $t,q$, the split interval size $\beta$, a split candidate $s$ in dimension $i$ and $\tilde n = |S|  + \text{Lap}(1/\varepsilon)$. Then, $|s|$ is the number of data points in $S^{(i)}$ that are contained in the split interval around $s$: $|s| = |\{x \in S^{(i)} \mid s - 0.5\beta \leq x \leq s + 0.5s\}|$. Also $r = r(s,S)$ is the rank of $s$ when inserted in the sorted $S^{(i)}$. We distinguish the cases that $s$ is in one of the outer quantiles $\outerQuantile = [0,\tilde nq]\cup[\tilde n-\tilde nq,\tilde n]$ and in the inner quantile $\innerQuantile=(\tilde nq, \tilde n-\tilde nq)$.
Then, the \emph{score} $f(S,\tilde n,s)$ is given by
\[
    \textstyle f_{t,q,\beta}(S,\tilde n,s) =  \alpha \cdot \bigg(\underbrace{ 1 - \frac{|s|}{\tilde n}}_{\textstyle e_\beta(S,\tilde n,s)}\bigg) + \underbrace{\begin{cases}
            \frac{(\frac{\tilde n}{2} - |r-\frac{\tilde n}{2}|) \cdot t}{\tilde nq} & \text{if $r \in \outerQuantile$}
            \\ \\
            \frac{t-2q}{1-2q} + \frac{(\frac{\tilde n}{2}-|r-\frac{\tilde n}{2}|) \cdot (1-t)}{\frac{\tilde n}{2}-\tilde nq} & \text{if $r \in \innerQuantile$}
        \end{cases}}_
    {c_{t,q}(S,\tilde n,s){s}} \text{.}
\]
If $t,q $ and $\beta$ are clear from context we write $f$. 
\end{definition}

\subsubsection{Exponential Mechanism}
The objective is to identify a split candidate that exhibits a high score in accordance with the specified scoring function. As DPM preserves DP, the Exponential Mechanism is employed as a noisy argmax algorithm. 
In order to guarantee the privacy of the data, noise is added to the selection process and thus only with high probability, a split with a high score is selected. Therefore, the utility of the selection process is diminished in comparison to the non-DP argmax.
The noise added to the selection process is scaled by the sensitivity of the scoring function, denoted as $\Delta_f$. It is essential to ensure that the sensitivity is not underestimated. Given that the sensitivity is in $O(n^{-1})$, it is necessary to ensure that the number of data points is overestimated in most cases. Consequently, the noisy count is shifted by an offset that determines the probability, denoted as $\delta$, that the privacy guarantees of the Exponential Mechanism are not upheld.
The utility guarantees of the Exponential Mechanism are presented in \cite{privacybook}. In \cite{DPM}, the utility guarantees were extended to encompass the case of selecting a good split with a maximum difference of $\omega$ to the optimal score.
\begin{theorem}[Exponential Mechanism's Utility - generalised (Theorem 5.2, \cite{DPM})]\label{thm:exponential_mechanism_utility}
Fixing a set $S\subseteq D \in \datasets$ and the set of candidates $W$, let $\omega \ge 0$, $OPT(S,f,W) = \max_{s\in W} f(S,\tilde n,s)\}$ and $W_{OPT_\omega} = \{s \mid |f(S,\tilde n,s) - OPT(S,f,W)| \le \omega\}$ denote the set of elements in $W$ which up to $\omega$ attain the highest score $OPT(S,f,W)$. Then, for some $\kappa>0$ the Exponential Mechanism $M_E$ satisfies the following property:
\begin{align*}
    &\textstyle\Pr\big[f(S,\tilde n,M_{E}(S,f,\varepsilon)) \le OPT(S,f,W) - \omega - \frac{2\Delta_{f}}{\varepsilon}\left(\ln\left( \frac{|W|}{|W_{OPT_\omega}|} \right)+ \kappa \right)\big] \le e^{-\kappa}\text{.}
\end{align*}
\end{theorem}

\subsubsection{Utility Guarantees}
% Utility guarantees shown in DPM
In \cite{DPM}, the authors present theoretical utility guarantees pertaining to the noisy count, the selection of suitable split candidates, and the halting criterion.

% Guarantees of noisy count
The discrepancy between the precise count and the noisy number of elements can be constrained by employing the cumulative density function of the Laplace mechanism and the offset. It is demonstrated that the probability of the noisy count diverging by a margin exceeding $\ln(\sqrt{n}/\delta)/\varepsilon$ is equal to $1/(2\sqrt{n})$.

% DPM finds splits with high emptiness
The guarantees pertaining to the assertion that DPM identifies splits with high emptiness encompass both the assurance that a $t'$-central split is selected and the guarantee that a $t'$-central split exhibits a high emptiness value. In the event that the selected split is $t'$-central, the discrepancy between its emptiness and that of the optimal split is shown to be asymptotically less than $O(\kappa/(n\varepsilon)) (1-t')/\alpha$ with probability $1-\exp(-\kappa)$. The probability of a given split being $t'$-central depends on the number of splits that are $t'$-central, which in turn depends on the value of $t'$, the size of the split interval, and the characteristics of the input data. Furthermore, the scores of these splits can be approximated by a lower bound and the sum of the scores of all split candidates by an upper bound. This guarantee can be understood by considering that a high impact of emptiness (large $\alpha$) reduces the probability for a $t'$-central split to be selected. As the value of $t'$ decreases, the number of $t'$-central splits increases, while the lower bound on their score decreases. Nevertheless, the overall probability that a $t'$-central split is selected still increases for smaller $t'$.

% DPM terminates appropriately
In a manner analogous to the bound on the probability that a $t'$-central split is selected, the authors posit that DPM halts when a split is selected that is not $t'$-central. With this general bound, $t'$ can be set to the largest centreness value such that the minimum cluster size would be violated if chosen.
\begin{lemma}[Outer quantile is chosen - {\cite[Lemma 5.6]{DPM}}]
    \label{lem:prob_innerQuantileSelected}
    Let $S\subseteq D$ be the current set with noisy count $\tilde n$ and $W$ the set of all split candidates and $W_{\ge t'}:= \{s \mid c(S,\tilde n,s) \ge t'\}$.
    Let $e_{\min} := \min_{s \in W} e(S,\tilde n,s)$ be the minimal emptiness over all splits $s \in W$. The score of every split candidate can be represented as $\alpha \cdot e_{\min} + t' + \ln a_s$ (for some $a_s \ge 1$). Then with $
        L_{\ge t'} = \textstyle\sum_{s \in W_{\ge t'}} a_s 
    $ and $
        \textstyle L_{< t'}  = \sum_{s \in W_{< t'}} a_s
    $, we know
    \begin{align*}
        &\Pr[s \in W_{< t'}] = 
        \frac{%
                1%
            }{%
                \frac{%
                    L_{< t'}%
                }{L_{\ge t'}} + 1%
            }
        \text{.}
    \end{align*}
\end{lemma}

\subsection{Research Questions}

% When does DPM halt?
The aforementioned guarantees of the stopping behaviour pertain solely to the scenario in which DPM terminates on the first recursion level (the first split is selected such that DPM halts). However, it can be demonstrated that even if a split is selected such that DPM does not halt immediately, we can bound the probability such that DPM halts on later recursion levels. In \cref{cha:stoppingCriterion}, we present a lower bound on the probability that DPM halts after $i$ splits. This is achieved by first providing a recursive description of the probability and then presenting an approximation of this probability for an arbitrary data set, as well as for an equally distributed and a Gaussian distributed data set.

% silhouette score
In \cref{cha:silhScore}, we discuss the implications of the hypothesis that DPM selects optimal splits for achieving a good clustering result in terms of the cluster metric silhouette score. Consequently, we examine the point-wise alteration of the silhouette score, specifically identifying which settings result in an improvement and which do not. Additionally, we assess the suitability of the silhouette score as a clustering metric, illustrating an instance where the silhouette score deteriorates following an optimal DPM-based split.

% Relation to theoretical approach of separability
Finally, in \cref{cha:xiRho}, we establish a link between the split candidates implemented in DPM and a more theoretical notion of $(\xi, \rho)$-separability. This enables us to comprehend the implications of a given emptiness on the associated separation and to assess the potential for identifying gaps that could be used to split the data set into clusters.

\subsection{Related Work}
This work presents a further analysis of the utility guarantees of the differentially private clustering algorithm DPM, as presented in \cite{DPM}.
The remainder of this section will distinguish between related work in the field of privacy-preserving clustering and the theoretical guarantees provided for the usefulness of an algorithm.

There is a line of work \cite{balkan,emmc,DPMaxCover,DPClEasyIns,LocDPClKMeans,DPClTightApprox} on differentially private clustering algorithms that provide theoretical guarantees regarding the optimal $k$-means clustering result. The work \cite{emmc} provides the best bounds in terms of inertia. In this work, we rather analyse the behaviour of the mechanism for different settings.
A line of research \cite{balkan,emmc,DPMaxCover,DPClEasyIns,LocDPClKMeans,DPClTightApprox} has been conducted on differentially private clustering algorithms that provide theoretical guarantees regarding the optimal $k$-means clustering result, the inertia. In \cite{emmc}, the authors provide the most robust bounds in terms of inertia. In this work, we rather consider a different clustering metric, the silhouette score, and additionally analyse the behaviour of the mechanism.

% Tree-based (similar approach) -> also theoretical
In \cite{PrivTree}, the authors present an algorithm for identifying heavy hitters while maintaining privacy. Their approach, which involves data-independent splits and the use of a scoring function as a stopping criterion, bears resemblance to DPM. However, the authors do not provide a theoretical analysis of the algorithm's utility. \\

% Similar algorithms
% LSH splits,
In \cite{DPM}, authors discuss the following privacy-preserving clustering algorithms in the context of DPM.
Prior to the development of DPM, the LSH splits algorithm, as presented in ~\cite{lshsplits2021}, represented the state-of-the-art of privacy-preserving clustering algorithms. The authors of LSH splits provided theoretical bounds on the error in terms of distance to the $k$-means clustering result, the inertia.

% Mondrian and Optigrid
The Mondrian algorithm, as described in \cite{mondrian}, and the Optigrid algorithm, as described in ~\cite{Optigrid99}, are both non-privacy-preserving algorithms that employ similar approaches to those used by DPM. Mondrian is not a clustering algorithm; rather, it aims to subdivide the data set into subsets of the same size. As with DPM, Mondrian recursively splits the input data set, with each split occurring along one axis. One key difference between the two algorithms is the manner in which the splits are selected. Mondrian splits at the median of the dimension with the highest variance, whereas DPM aims to split close to the median but also considers the density of a split environment. 
Optigrid is a clustering algorithm that considers the dimensions of the input data set separately in order to deal with high-dimensional data. The aim is to find splits through sparse regions, which is similar to the goal of DPM. This is achieved by using kernel density estimation (per dimension).
In order to analyse the utility of Optigrid, the other authors discuss the expected behaviour of the algorithm for different input distributions. They also discuss the error in the worst case. In this work, we follow a similar approach to a theoretical analysis of the utility of a mechanism. Furthermore, we investigate the impact of different hyperparameters on the performance and behaviour of the mechanism.

% Similar proof
In \cite{vempala2012randomly}, the authors provide a refinement of random projection trees that adapt to intrinsic dimension by adding a random rotation  as a pre-processing step. To substantiate this assertion, the authors present a rationale for the quality of selected splits and the probability of a favourable outcome. This proof sketch is analogous to the utility proofs presented in \cite{DPM} and this work. The principal distinctions are that \cite{vempala2012randomly} demonstrates that the diameter of each subset is diminished after a sufficient number of splits. However, this is not applicable to DPM, as the number of splits cannot be arbitrarily large, as each split results in the loss of information and an exponential increase in the number of clusters. Moreover, the objective of clustering is not merely to reduce the diameter of subsets; rather, it is to avoid splitting up clusters and to identify similarities in the data points.

%% file: main/stoppingCriterion.tex
\section{Stopping Criterion of DPM}
\label{cha:stoppingCriterion}

% Problem of current approach
In \cite{DPM}, the authors put forth the argument that DPM halts with the probability that a split candidate is selected in a way that violates the minimum number of elements. 

Furthermore, they provide bounds for this probability. This approach considers solely the possibility that a split is selected in a way that causes DPM to immediately halt. However, this approach does not account for the possibility that DPM may halt within the subsequent recursion steps, which could still be considered as halting appropriately. Therefore, we present an inductive probability for the occurrence of DPM halting within a fixed number of steps.

% How are we planning to do that?
The probability that DPM halts within a fixed number of recursion levels $i$, can be expressed as the sum of the probabilities that DPM halts after each recursion level $i\le j$. For the initial level($i=0$), the probability that DPM halts is contingent upon the probability of a split being selected that results in the minimum cluster size being violated. For all subsequent split levels, the probability that DPM halts after the $i+1$ level is equal to the probability of DPM halting after the $i+1$ split level, plus the probability of DPM halting after $i$ levels. With each split, and thus each level, the number of data points in the considered set is reduced. As the number of data points decreases, the likelihood of splits violating the minimum cluster size increases. If the selected split on level $i$ has a high centreness, the number of data points is decreased significantly and thus for level $i+1$ the minimum cluster size is violated for larger emptiness values.

% Lemma and Proof
%% General idea
The centreness threshold ,$t_\tau$, is assigned such that all split candidates with lower centreness violate the minimum cluster size.
Given that the centreness function is comprised of two linear functions, it is necessary to consider two distinct cases with respect to the centreness threshold. In the first case, selected splits that violate the minimum cluster size are located in the outer quantile, while in the second case, selected splits that violate the minimum cluster size are located in the inner quantile. The set of all split candidates with centreness larger or equal to $t_\tau$ is defined as $W_{\ge t_\tau}$ (analogously $W_{<t_\tau}$) and for a centreness larger than $t_\tau$ and less than $t'$, we write $W_{> t_\tau\land <t'}$. 
% Notations
Given a set $S$, we denote the partition of $S$ into the subsets $S'$ and $S''$ (with $S'$ and $S''$ being disjoint) as $S = S' \dot \cup S''$. With the centreness metric as defined in \cref{def:score}, the centreness of a partition can be determined (if split according to DPM). Then, we write $S = S' \dot \cup_\theta S''$.

\subsection{General Stopping Guarantees}
In general, we can write the probability that DPM halts as a conditional probability. Over all possible centreness values, we have to multiply the probability that this value occurs with the probability that DPM halts assuming a given centreness value (\cref{eq:prDPMhalts}). We distinguish the cases that DPM halts immediately without further splitting and that DPM halts on a later recursion level (\cref{eq:prDPMhaltsRec}). In the latter, we have to ensure that DPM halts for both subsets by multiplying the probabilities for each subset as well as the probability that such a separation occurs. 
\begin{align}
    &\Pr[\text{DPM halts on }S] \nonumber\\
    =& \sum_{\theta\in[0,1]}\Pr[\text{DPM halts on }S|t'=\theta]\cdot \Pr[t'=\theta] \label{eq:prDPMhalts} \\
   =& \sum_{\theta\in[0,t_\tau]} \Pr[t'=\theta] + \sum_{\theta\in[t_\tau, 1]}\Pr[\text{DPM halts on }S|t'=\theta]\cdot \Pr[t'=\theta]\label{eq:prDPMhaltsRec} \\
   =& \sum_{\theta\in[0,t_\tau]} \Pr[t'=\theta]\nonumber \\ & + \sum_{\theta\in(t_\tau, 1]} \sum_{\substack{S',S'':\\ S'\dot \cup_\theta S'' = S}} \mkern-18mu \Pr[\underbrace{S \text{ split into } S',S''}_{=A}] \Pr[\text{DPM halts on }S'|A] \cdot \Pr[\text{DPM halts on }S''|A] \label{eq:prDPMhaltsRecLR} 
\end{align}
In order to gain further insight into the probability of DPM halting, we begin by calculating the probability that DPM halts immediately on a given set, denoted by $S$. This is (simplified) defined as the ratio of all splits that cause DPM to halt and all possible splits (first summand in \cref{eq:prDPMhaltsRec}) which was already provided in \cite{DPM}.
Secondly, we analyse the probability of DPM halting at deeper recursion levels. To do so, we consider all possible partitions for every centreness value and calculate the probability of each partition occurring(second summand in \cref{eq:prDPMhaltsRecLR}).

\subsubsection{First Split Level}
%% First level
The probability that DPM halts without any further splitting (first summand in \cref{eq:prDPMhaltsRec}) is non-zero if the centreness value of a split is below the threshold $t_\tau$. Thus, the probability that DPM halts immediately is determined by the probability that a split is selected that violates the minimum cluster size $\tau_e$ and thus as given in \cite{DPM}.

\begin{lemma}[DPM halts immediately]\label{lem:haltImmediately}
Given a set $S$, the noisy count of elements in $S$ $\tilde n_S$, the minimum size of elements, the set of all split candidates $W$, score function $f$ as defined in \cref{def:score} with valid centreness parameters ($t\ge 2q$) and sensitivity $\Delta_f$,  Then, the probability that DPM as described in \cref{sec:dpm} halts immediately on $S$ is
    \begin{align}
        &\Pr[\text{DPM halts immediately on }S] \nonumber \\ 
        =& \sum_{\theta \in [0,t_\tau]}\Pr[t' = \theta] = \Pr[s \in W_{\le t_\tau}] \nonumber \\
        =& \frac{\sum_{s\in W_{\le t_\tau}}\exp({f(S,\tilde n, s)\varepsilon/(2\Delta_f)})}{\sum_{s \in W}\exp{(f(S,\tilde n, s)\varepsilon/(2\Delta_f)})} \text{ and } t_\tau = \begin{cases}
        \frac{(\frac{\tilde n}{2} - |\tau_e-\frac{\tilde n}{2}|) \cdot t}{\tilde nq} & s\in Q_O \\
         \frac{t-2q}{1-2q} + \frac{(\frac{\tilde n}{2}-|\tau_e-\frac{\tilde n}{2}|) \cdot (1-t)}{\frac{\tilde n}{2}-\tilde nq} & s \in Q_I
        \end{cases} \label{eq:prDPMhaltsImmediate}
    \end{align}
    where $\outerQuantile = [0,\tilde nq]\cup[\tilde n-\tilde nq,\tilde n]$ and $\innerQuantile=(\tilde nq, \tilde n-\tilde nq)$.
\end{lemma}
\begin{proof}
    The probability that DPM halts immediately on a given set $S$ is equal to the probability of selecting a split that violates the minimum cluster size. The numerator in \cref{eq:prDPMhaltsImmediate} represents the sum of the exponential scores associated with all splits that violate the minimum cluster size. The centreness threshold $t_\tau$ depends on the minimum cluster size and the size of the set $S$.  Furthermore, the numerator is influenced by the scores of the splits within the split candidate set, $W_{\le t_\tau}$.
    The denominator is the sum of the exponential scores of all potential split candidates and thus is solely contingent on the scores of the split candidates for a given set, $S$.  
\end{proof}

We discussed the case that DPM halts immediately on a set $S$ (similar to \cite{DPM}). As DPM recursively divides a set, it remains to show the probability that DPM halts before level $j$. As an induction step, we first analyse the probability that DPM halts on the next recursion level. 

%% induction (i+1)
\subsubsection{Induction Step}
It is possible that DPM also halts on later levels, in addition to immediately. Therefore, the probability that DPM halts on the next level for both subsets should be added to the probability that DPM has already halted.
The probability that DPM halts on the subsequent level $i+1$ is contingent upon the entirety of the splits that have occurred at preceding levels, in particular to the resulting partitions. Consequently, when considering a given set $S$, all splits that do not immediately cause DPM to halt are taken into account. Rather than considering all potential split candidates, we instead examine partitions corresponding to specified centreness values, denoted by $\theta$. For all $\theta$ larger than the centreness threshold $t_\tau$, we go through all partitions of $S$ resulting in $\theta$. Subsequently, we multiply the probability of this partition (2) and the probability that DPM halts on the resulting subsets, denoted by $S'$ and $S''$, respectively (1).

\begin{lemma}[DPM halts on next recursion level] \label{lem:haltNextRecLevel}
Given a set $S$, the corresponding centreness threshold for $S$. Assuming that $j<\tau_s$ and even further $i$ less than the remaining recursion level, the probability that DPM as described in \cref{sec:dpm} halts on the next recursion level is as follows
    \begin{align}
        &\Pr[\text{DPM halts on }S \text{ on next recursion level}] \nonumber\\
        &=\sum_{\theta\in(t_\tau, 1]} \sum_{\substack{S',S'':\\ S'\dot \cup_\theta S'' = S}} \mkern-18mu  \underbrace{\Pr[A=S \text{ split into } S',S'']}_{(2)} \cdot \underbrace{\Pr[\text{DPM halts on }S' \text{ and on }S''|A]}_{(1)} \text{.}\label{eq:prDPMhaltsRecLRnum}
    \end{align}
\end{lemma}
\begin{proof}
    We know that all splits with a centreness score of less than or equal to $t_\tau$ cause DPM to halt immediately. Therefore, only splits or resulting partitions with a centreness value exceeding $t_\tau$ are considered. For each partition, the probability that DPM halts for the resulting subsets (1) is multiplied by the probability that this partition occurs (2). 
    All splits for which DPM does not halt are considered, and the probability of DPM halting on the subsequent recursion level if this split occurs is summed, with each occurrence weighted according to its probability.
\end{proof}

\begin{remark}
    For the last considered recursion level, we know that only the probability that DPM halts immediately is relevant as there is no next recursion level and thus \[\sum_{\theta\in[t_\tau, 1]}\Pr[\text{DPM halts on }S|t'=\theta]\cdot \Pr[t'=\theta] = 0 \text{.}\]
\end{remark}

The probability that DPM halts on $S'$ and $S''$ depends on the ratio of the scores of split candidates that violate the minimum cluster size and the scores of all split candidates. The aforementioned split candidates that contravene the minimum cluster size are defined by the centreness threshold, $t_\tau$. The centreness threshold, $t_\tau$, is solely dependent on the size of the current set, with the fixed centreness parameters, $t$ and $q$, held constant. Despite the fixed minimum cluster size, the threshold is observed to increase for smaller sets. Accordingly, the centreness threshold with respect to a given set $S$, $t_\tau(S)$, is defined.
. By considering the values of $t_\tau(S')$ and $t_\tau(S'')$, where $S'$ and $S''$ represent the partitioning of $S$, the probability that DPM halts on the subsequent recursion level can be calculated using the equation provided in \cref{eq:prDPMhaltsImmediate}.

The preceding steps yield the following recursive equation for the probability that DPM halts after $i$ splits on a given set $S$.
\begin{theorem}
Given a set $S$, the set of split candidates $W$, the centreness threshold $t_\tau(S)$ and DPM as described in \cref{sec:dpm}.
Then, the probability that DPM halts on $S$ after at most $i$ steps is recursively given by
    \begin{align}
        &\Pr[\text{DPM halts after }i \text{ splits on }S]\nonumber\\
        \stackrel{(*)}{=}&\underbrace{\Pr[\text{DPM halts immediately on }S]}_{\cref{lem:haltImmediately}} + \underbrace{\Pr[\text{DPM halts on }S\text{ on next recursion level}]}_{\cref{lem:haltNextRecLevel}} \label{eq:recIthSplit}\\
        =&\Pr[s\in W_{\le t_\tau(S)}] \nonumber\\
        &+ \sum_{\theta\in[t_\tau(S), 1]} \sum_{\substack{S',S'':\\ S'\dot \cup_\theta S'' = S}}\mkern-18mu \Pr[S \text{ split into } S',S''] \underbrace{\Pr[\text{DPM halts after }i-1 \text{ splits on }S' \text{ and }S'']}_{\substack{= \Pr[\text{DPM halts after }i-1 \text{ splits on }S']\\\cdot \Pr[\text{DPM halts after }i-1 \text{ splits on }S''|A] }} \label{eq:finalProbDPMhaltsI}
    \end{align}
\end{theorem}
\begin{proof}
    In order for DPM to halt on a given data set, $S$, it must either halt immediately or halt on the next recursion levels, provided that the maximum recursion depth has not yet been reached. Subsequently, \cref{eq:finalProbDPMhaltsI}is derived with the aid of \cref{lem:haltImmediately} and \cref{lem:haltNextRecLevel}. The probability that DPM halts after $i-1$ splits on $S'$ and $S''$ is the product of the probabilities for $S'$ and $S''$ under the assumption that this partition occurs. These probabilities can then be determined by plugging in $S'$ or $S''$ for $S$ and $i-1$ for $i$ in $(*)$.
\end{proof}
As the algorithm DPM itself, we can define the probability that DPM halts until a given recursion level is reached. This probability can be approximated by making assumptions about the emptiness of the inner quantile, as well as the minimum occurring emptiness. Consequently, a lower bound for the probability that DPM halts can be derived, which depends on the hyperparameters and assumptions about the input data set. If it can be assumed that there is always a central split, even tighter bounds can be obtained. 

\subsubsection{Approximation of Lower Bounds}
The approximation of the probability that DPM halts, heavily depends on the input data set.
Furthermore, the following assumptions are made: the emptiness in the inner quantile is at most $e_{Q_I}$, as otherwise the cluster would still be clusterable; as in \cite{DPM}, it is assumed that all split candidates have an emptiness of at least $e_{\text{min}}$.

% Approximation of immediate halting
Firstly, we establish a lower bound on the probability that DPM will halt immediately when presented with a given set, $S$, and the corresponding centreness threshold, $t_\tau(S)$.
\begin{lemma}[Lower bound for DPM to immediately halt]
    \label{lem:lowBoundHaltImm}
    Given a set $S$ and its noisy count $\tilde n$, and corresponding centreness threshold $t_\tau(S)$ as well as a set of split candidates $W$. Further, given DPM as introduced in \cref{sec:dpm} with scoring function $f$, metric weight $\alpha$, centreness parameter $t$ with privacy budget $\varepsilon$ for the Exponential Mechanism, and sensitivity $\Delta_f$. Assuming that the minimum occurring emptiness is $e_{\min}$ and the maximum emptiness in the inner quantile is $e_{Q_I}$, the probability that DPM  halts on $S$ immediately can be lower bounded as follows
    \begin{align*}
        &\Pr[\text{DPM immediately halts on } S] = \Pr[s\in W_{\le t_\tau(S)}]\\
        \ge& \frac{|W_{\le t_\tau(S)}| e^{\emin\varepsilon/(2\Delta_f)}}{|W_{\le t_{\tau(S)}}|e^{(t_{\tau(S)}+\alpha)\varepsilon/(2\Delta_f)}+ |W_{\ge t}|e^{(1 +\alpha e_{Q_I})\varepsilon/(2\Delta_f)}+ |W_{>t_{\tau}(S)\land <t}|e^{(t+\alpha)\varepsilon/(2\Delta_f)}}
    \end{align*}
\end{lemma}
\begin{proof}
    \begin{align*}
        &\Pr[s\in W_{\le t_\tau(S)}] \\
        =&  \frac{\sum_{s\in W_{\le t_\tau(S)}}\exp({f(S,\tilde n, s)\varepsilon/(2\Delta_f)})}{\sum_{s\in W}\exp{(f(S,\tilde n, s)\varepsilon/(2\Delta_f)})}\\
        \intertext{We can give a lower bound on the nominator $\sum_{s\in W_{\le t_\tau(S)}} \exp(f(S,\tilde n, s)\varepsilon/(2\Delta_f))$. We only know that splits $W_{\le t_\tau}$ that violate the centreness threshold have at least a score of $\emin$. Thus, $\sum_{s\in W\le t_\tau(S)} \exp(f(S,\tilde n, s)\varepsilon/(2\Delta_f))\ge |W\le t_\tau(S)| \exp(\emin\varepsilon/(2\Delta_f))$. We can also give an upper bound on the denominator $\sum_{s\in W}\exp{(f(S,\tilde n, s)\varepsilon/(2\Delta_f)})$. For the sum of all split candidate we distinguish the split candidates that violate the centreness threshold (1), the split candidates that are in the inner quantile (2) and all other split candidates (3). Then, $\sum_{s\in W}\exp{(f(S,\tilde n, s)\varepsilon/(2\Delta_f)}) \le (1) + (2) + (3)$ and $(1)\le |W_{\le t_{\tau(S)}}|\exp((t_{\tau(S)}+\alpha)\varepsilon/(2\Delta_f))$, $(2)\le |W_{\ge t}|\exp((1 +\alpha e_{Q_I})\varepsilon/(2\Delta_f))$ and $(3)\le |W_{>t_{\tau(S)}\land <t}|\exp((t+\alpha)\varepsilon/(2\Delta_f))$.}
        \ge& \frac{|W_{\le t_\tau(S)}| e^{\emin\varepsilon/(2\Delta_f)}}{|W_{\le t_{\tau(S)}}|e^{(t_{\tau(S)}+\alpha)\varepsilon/(2\Delta_f)}+ |W_{\ge t}|e^{(1 +\alpha e_{Q_I})\varepsilon/(2\Delta_f)}+ |W_{>t_{\tau(S)}\land <t}|e^{(t+\alpha)\varepsilon/(2\Delta_f)}}
    \end{align*}
\end{proof}

In order to provide a lower bound on the probability that DPM halts on later recursion levels, two different cases are considered. If it can be assumed that a $t'$-central split is always present, all possible partitions can be reduced to $t'$-central splits, allowing the centreness threshold for the next recursion level to be approximated accordingly. Given that there are no splits with high emptiness in proximity to the centre by assumption, this assumption is reasonable for suitable split interval sizes and ranges. However, as the current implementation of DPM does not adapt the split intervals to the current subset, and thus it is not realistic to assume that there is always a $t'$-central split. Consequently, the number of potential partitions is reduced to those that do not result in DPM halting. With the additional assumption, the resulting bounds are more precise, as even in the most unfavourable scenario of a $t'$-central split, the centreness is limited to $t'$. This implies a more balanced partition and, consequently, a higher centreness threshold for the subsequent recursion level.  In certain settings, characterised by large values of $t'$ and  relatively large split intervals, there might be no $t'$-central split. In such cases, the probability of DPM halting at a later recursion level is likely to be overestimated. We consider the probability with the additional assumption and without, as with additional background knowledge and different implementations of DPM, the tighter lower bounds on the probability are likely to be accurate.

%% Case 1
\paragraph{There is a $t'$-central split} 
\label{Sssec:tPrimeCentral}
Assume that there is always a $t'$-central split.
Then, we can give a lower bound on the probability for all $\theta \ge t_\tau$  that a partition occurs ($\Pr[S \text{ split into } S',S'']$) by lower bounding the probability that a $t'$-central split is selected.

\begin{lemma}[Lower bound on probability for $t'$-central split]
\label{lem:lowerBoundCentral_untight}
Given a set $S$ and its noisy count $\tilde n$, and corresponding centreness threshold $t_\tau(S)$ as well as a set of split candidates $W$. Further, given DPM as introduced in \cref{sec:dpm} with scoring function $f$, metric weight $\alpha$, centreness parameter $t$ with privacy budget $\varepsilon$ for the Exponential Mechanism, and sensitivity $\Delta_f$. Assuming that $t_\tau<t'<t$ and there is at least one $t'$-central split, the minimum occurring emptiness is $e_{\min}$ and the maximum emptiness in the inner quantile is $e_{Q_I}$, the probability that a $t'$-central split is selected can be lower bounded as follows:
    \begin{align*}
        &\Pr[s\in W_{\ge t'}]\\
        \ge& \frac{\exp((t_\tau(S) + \emin \alpha)\varepsilon/(2\Delta_f))}{|W_{\le t_{\tau(S)}}|e^{(t_{\tau(S)}+\alpha)\varepsilon/(2\Delta_f)}+ |W_{\ge t'}|e^{(1 +\alpha e_{Q_I})\varepsilon/(2\Delta_f)}+ |W_{>t_{\tau}(S)\land <t'}|e^{(t'+\alpha)\varepsilon/(2\Delta_f)}}
    \end{align*}
\end{lemma}
\begin{proof}
The probability that a $t'$-central split is selected is determined by the ratio of the sum of all scores with centreness above $t'$ compared to the sum of scores of all split candidates.
\begin{align*}
     \Pr[s\in W_{\ge t'}] &= \frac{\sum_{s\in W_{\ge t'}}\exp({f(S,\tilde n, s)\varepsilon/(2\Delta_f)})}{\sum_{s\in W}\exp{(f(S,\tilde n, s)\varepsilon/(2\Delta_f)})}
 \end{align*}
We can lower bound the nominator $\sum_{s \in W_{\ge t_\tau(S)}} \exp(f(S,\tilde n, s)\varepsilon/(2\Delta_f))$. We know that $t'$-central splits $W_{\ge t'}$ have at least a centreness of $t'$ and emptiness of $\emin$, thus we get $\sum_{s\in W\ge t'} \exp(f(S,\tilde n, s)\varepsilon/(2\Delta_f))\ge \exp((t' + \emin \alpha)\varepsilon/(2\Delta_f))$. To get an upper bound on the denominator, i.e. the scores of all split candidates $\sum_{s\in W}\exp{(f(S,\tilde n, s)\varepsilon/(2\Delta_f)})$, we consider the following worst-case. We distinguish the following three sets of split candidates: (1) $t'$-central split candidates have a centreness of at most $1$ and an emptiness of at most either $1$ or $e_{Q_I}$ if $t'>t$. (2) The splits with centreness larger than $t_\tau$ but less than $t'$ with a maximum emptiness of $1$. (3) We are left with the split candidates that cause DPM to halt with upper bounded centreness of $t_\tau$ and emptiness of $1$. Finally, we can put this together to upper bound the denominator as follows: $\sum_{s\in W}\exp{(f(S,\tilde n, s)\varepsilon/(2\Delta_f)}) \le (1) + (2) + (3)$ and $(1) \le |W_{\le t_{\tau(S)}}|\exp((t_{\tau(S)}+\alpha)\varepsilon/(2\Delta_f))$, $(2)\le |W_{>t_{\tau(S)}\land <t'}|\exp((t'+\alpha)\varepsilon/(2\Delta_f))$ and $(3)\le |W_{\ge t'}|\exp((1 +\alpha e_{Q_I})\varepsilon/(2\Delta_f))$.
\end{proof}

In considering the subsequent recursion step, it is now feasible to limit the scope to $t'$-central splits. This may be approximated by determining the probability of the partition exhibiting the least favourable halting characteristics. It is established that the centreness threshold, $t_\tau$, increases with each successive  split. The magnitude of the increase is contingent upon the ratio of the subset size prior to and subsequent to the split. Thus, the initial step is to ascertain the alteration in $t_\tau$. For a given value of $t'$, the size of the sets resulting from the split can be determined.  It is necessary to consider two cases: firstly, where $t'>t$, and secondly, where $t'\le t$. The partitioning of a set $S$ into $S'$ and $S''$ when $t'>t$ (inner quantile) is as follows with regard to the size: 
\begin{align*}
    &t' = \frac{(\frac{\tilde n}{2}-|r-\frac{\tilde n}{2}|)t}{\tilde n q} \\
    \leftrightarrow \quad & \frac{t'}{t}\tilde nq - \frac{\tilde n}{2} = |r-\frac{\tilde n}{2}|\\
    \leftrightarrow \quad & |S'| = \frac{t'}{t}\tilde n q, |S''| = \tilde n -\tilde nq\frac{t'}{t}
\intertext{If $t'\le t$ (outer quantile), the partitioning regarding the size is as follows: }
    &t' = \frac{t-2q}{1-2q}+\frac{(\frac{\tilde n}{2}-|r-\frac{\tilde n}{2})(1-t)}{\frac{\tilde n}{2}-\tilde n q}\\
    \leftrightarrow \quad & \big(t'-\frac{t-2q}{1-2q}\big) \frac{\frac{\tilde n}{2}-\tilde n q}{1-t} = \frac{\tilde n}{1}|r-\frac{\tilde n}{2}|\\
    \leftrightarrow \quad &|S'| = \big(t'-\frac{t-2q}{1-2q}\big)\frac{\frac{\tilde n}{2}-\tilde nq}{1-t}, |S''| = \tilde n -\big(t'-\frac{t-2q}{1-2q}\big)\frac{\frac{\tilde n}{2}-\tilde nq}{1-t}
\end{align*} The next step is to plug in the size of the partitions $S'$ and $S''$, in order to determine the corresponding $t_\tau(S')$ and $ t_\tau(S'')$. We set $t=2q$ which allows us to use $c_{2q,q}(S,s,\tilde n_S) = 1 - |\frac{2r}{\tilde n_S } - 1|$ as a lower bound on the centreness function, as only $t\ge 2q>0$ are valid parameters. Therefore, for any given values of $\tau_e$ and $\tilde n_S$, the lower bound on the centreness threshold for any $t,q$ is given by the centreness threshold for $c_{2q,q}$: $t_\tau(S) = 1 - |\frac{2r}{\tilde n_S } - 1|$. It can be demonstrated that $2\tau_e \le \tilde n_S$. This is because, if this were not the case, it would follow that there is no split candidate for which the minimum cluster size is satisfied. Therefore, it can be concluded that for both cases, namely $r=\tau_e$ and $r=\tilde n_S-\tau_e$,the value of $t_\tau(S)= \frac{2\tau_e}{\tilde n_S}$. 
In order to ascertain the extent of improvement in $t_\tau$, it is necessary to consider the various scenarios pertaining to the relationship between $t'$ and $t$.
\begin{align*} 
    &(1)\quad t'>t\to |S'| = \frac{t'}{t}\tilde n_S q, |S''| = \tilde n_S -\frac{t'}{t}\tilde n_S q \text{ and } \frac{t'}{t} q < 1\\
    &\quad t_\tau(S') \ge \frac{2\tau_e}{\frac{t'}{t}\tilde n_S q } 
    = \underbrace{\frac{2\tau_e}{\tilde n_S}}_{t_\tau(S)} \cdot \underbrace{\frac{t}{t'q}}_{>1} > t_\tau(S)\\
    &\quad t_\tau(S'') \ge \frac{2\tau_e}{\tilde n_S -\frac{t'}{t}\tilde n_S q}
    = \underbrace{\frac{2\tau_e}{\tilde n_S}}_{t_\tau(S)} \cdot \underbrace{\frac{1}{1-\underbrace{\frac{t'q}{t}}_{<1}}}_{>1} > t_\tau(S)\\
    &(2)\quad t'\le t\to |S'| = \bigg(\frac{t-2t'-2q}{2(t-1)}\bigg)\tilde n_S, |S''| = \tilde n_S -\bigg(\frac{t-2t'-2q}{2(t-1)}\bigg)\tilde n_S, \bigg(\frac{t-2t'-2q}{2(t-1)}\bigg) < 1\\
    & \quad t_\tau(S') \ge \frac{2\tau_e}{\big(\frac{t-2t'-2q}{2(t-1)}\big)\tilde n_S}
    = \underbrace{\frac{2 \tau_e}{\tilde n_S}}_{t_\tau(S)} \cdot \underbrace{\frac{1}{\frac{t-2t'-2q}{2(t-1)}}}_{>1}>t_\tau(S)\\
    & \quad t_\tau(S'') \ge \frac{2\tau_e}{\tilde n_S-\big(\frac{t-2t'-2q}{2(t-1)}\big)\tilde n_S} 
    = \underbrace{\frac{2 \tau_e}{\tilde n_S}}_{t_\tau(S)} \cdot\underbrace{\frac{1}{1- \underbrace{\frac{t-2t'-2q}{2(t-1)}}_{>1}}}_{>1} > t_\tau(S)
\end{align*}
Four distinct alterations to $t_\tau$ were contemplated for a subset, contingent on a given $t'$. However, the inequality $t'\le t$ or $t'>t$ remains unchanged for different subsets and split levels, given that $t'$ is fixed and $t$ is a hyperparameter of DPM. In both cases (1) and (2), we determine $t_\tau(S')$ and $t_\tau(S'')$. For the sake of simplicity, we will only consider the most unfavourable scenario, which corresponds to the smaller centreness threshold which results from the larger subset. This results in a less precise analysis, as the probability of DPM terminating is underestimated for smaller subsets. Whether $|S'|>|S''|$ or $|S'|<|S''|$ depends on $t',t,q$, with $t'$ being fixed and $t,q$ being hyperparameters. Consequently, we can approximate $t_\tau$ on split level $i$ with $S$ as the origin set as follows.

\begin{align}
    t^j_\tau(S) &\ge \frac{2\tau_e}{\tilde n_S} \cdot \begin{cases}
    \min\bigg(\frac{t}{t'q}, \frac{1}{1-\frac{t'q}{t}}\bigg)^i,   & t'>t \\
    \min \bigg(\frac{1}{\frac{t-2t'-2q}{2(t-1)}}, \frac{1}{1- \frac{t-2t'-2q}{2(t-1)}}\bigg)^i, & t' \le t
    \text{.}\end{cases}\label{eq:centThreshold}
\end{align}

As previously stated, the value of $t'$ remains constant for all levels of recursion. However, the rank at which a split candidate achieves centreness of at least $t'$ varies with each split. For the sake of simplicity, we can also express $t'$ in terms of the split level. This merely necessitates interpreting the $t'$ values of subsequent levels as a centreness value for the origin set $S$. 

\begin{align}
    t'^\ell &\ge t' \cdot \begin{cases}
    \min\bigg(\frac{t}{t'q}, \frac{1}{1-\frac{t'q}{t}}\bigg)^\ell,   & t'>t \\
    \min \bigg(\frac{1}{\frac{t-2t'-2q}{2(t-1)}}, \frac{1}{1- \frac{t-2t'-2q}{2(t-1)}}\bigg)^\ell, & t' \le t
    \text{.}\end{cases}\label{eq:TPrime}
\end{align}

The aforementioned building blocks allow us to establish a lower bound on the probability that DPM halts after at most $i$ splits. It should be recalled that we are operating under the assumption that there is a $t'$-central split.

\begin{theorem}[Lower bound on probability that DPM halts after $i$ splits ($t'$-splits)]
\ \ \\Given a set $S$ and its noisy count $\tilde n$, and corresponding centreness threshold $t_\tau(S)$ as well as a set of split candidates $W$. Further, given DPM as introduced in \cref{sec:dpm} with scoring function $f$ with sensitivity $\Delta_f$, metric weight $\alpha$, centreness parameters $t,q$ and privacy budget $\varepsilon$ for the Exponential Mechanism. Assuming that the minimum occurring emptiness is $e_{\min}$ and the maximum emptiness in the inner quantile is $e_{Q_I}$,
Then, for a fixed $t'$ the probability that DPM halts until recursion level $i$ can be lower bounded as follows.
   \begin{align}
   \label{eq:finalProbHaltTprime_untight}
        &\Pr[\text{DPM halts after }i \text{ splits on }S] \\\ge&
        \sum^j_{i = 0} \underbrace{\Pr[s\in W_{\le t_\tau^i(S)}]}_{\cref{lem:lowBoundHaltImm}} \cdot \bigg(\prod^j_{\ell = 0} \underbrace{\Pr[s\in W_{\ge t'^\ell}]}_{\cref{lem:lowerBoundCentral_untight}}\bigg)^{2^i}\\
         = &\sum^j_{i = 0} \frac{|W_{\le t^j_\tau(S)}| e^{\emin\varepsilon/(2\Delta_f)}}{|W_{\le t_{\tau(S)}}|e^{(t^j_{\tau}(S)+\alpha)\varepsilon/(2\Delta_f)}+ |W_{\ge t}|e^{(1 +\alpha e_{Q_I})\varepsilon/(2\Delta_f)}+ |W_{>t^j_{\tau}(S)\land <t}|e^{(t+\alpha)\varepsilon/(2\Delta_f)}} \\
         & \mkern-6mu\bigg(\mkern-2mu \prod^j_{\ell = 0}\frac{\exp((t^\ell_\tau(S) + \emin \alpha)\varepsilon/(2\Delta_f))}{|W_{\le t^\ell_{\tau(S)}}|e^{(t^\ell_{\tau}(S) \mkern-2mu+\mkern-2mu\alpha)\varepsilon/(2\Delta_f)}\mkern-2mu+ \mkern-2mu|W_{\ge t'^\ell}|e^{(1 +\alpha e_{Q_I})\varepsilon/(2\Delta_f)} \mkern-2mu + \mkern-2mu |W_{>t^\ell_{\tau}(S)\land <t'^\ell}|e^{(t'^\ell+\alpha)\varepsilon/(2\Delta_f)}}  \mkern-5mu \bigg)^{2^i}\\
        &\text{ with }t^j_\tau(S) = \frac{2\tau_e}{\tilde n_S} \cdot c \text{ and }t'^\ell = t' \cdot c \text{ where } c = \begin{cases}
        \min\bigg(\frac{t}{t'q}, \frac{1}{1-\frac{t'q}{t}}\bigg)^i,   & t'>t \\
        \min \bigg(\frac{1}{\frac{t-2t'-2q}{2(t-1)}}, \frac{1}{1- \frac{t-2t'-2q}{2(t-1)}}\bigg)^i, & t' \le t \text{.}
        \end{cases} 
    \end{align} 
\end{theorem}
\begin{proof}
    The probability that DPM halts immediately for a given set $S$ and a recursion level $i$ can be lower bounded by approximating the centreness threshold $t^j_ \tau(S)$ as in \cref{eq:centThreshold} and interpreting $t'$ on level $\ell$ as $t'^\ell$ according to \cref{eq:TPrime}. Subsequently, the probabilities that DPM halts without further splitting can be calculated using the results of the previous section and the lower bound given in \cref{lem:lowBoundHaltImm}. In order to calculate the probability that DPM does not halt but instead performs a $t'$-central split, we apply the lower bound given in \cref{lem:lowerBoundCentral_untight}. For DPM to halt on some recursion level $i$, it has to halt for all $2^i$ subsets on this level. Thus, we must take the product of all these possibilities. 
\end{proof}

%% Case 2
\paragraph{There is no $t'$-central split}
If we do not know for sure that there is a $t'$-central split, the lower bound on the minimum centreness threshold as in \cref{eq:centThreshold} does not hold.

First, as in \cref{Sssec:tPrimeCentral} we determine the probability that a $t_\tau(S)$-split is selected for some set $S$. Then, we again lower bound the centreness threshold for some split level.

\begin{lemma}[Lower bound on probability that DPM does not halt]
\label{lem:lowerBoundNotHalt_untight}
Given a set $S$ and its noisy count $\tilde n$, and corresponding centreness threshold $t_\tau=t_\tau(S)$ as well as a set of split candidates $W$. Further, given DPM as introduced in \cref{sec:dpm} with scoring function $f$, metric weight $\alpha$, centreness parameter $t$ with privacy budget $\varepsilon$ for the Exponential Mechanism, and sensitivity $\Delta_f$. Assuming the minimum occurring emptiness is $e_{\min}$, the probability that DPM does not halt on $S$ can be lower bounded as follows
    \begin{align*}
        &\Pr[s\in W_{> t_{\tau}}]\\
        &\ge 1-\frac{|W_{\le t_{\tau}}| e^{(\alpha + t_\tau)\varepsilon/(2\Delta_f)}}{|W_{\le t_{\tau}}|e^{\emin\alpha \varepsilon/(2\Delta_f)}+ |W_{> t_{\tau}, \le t}|e^{(\emin\alpha + t_\tau)\varepsilon/(2\Delta_f)} + |W_{\ge t, > t_\tau}|e^{(\emin\alpha + t)\varepsilon/(2\Delta_f)}}
     \end{align*}
\end{lemma}
\begin{proof}
The probability that a split is selected such that DPM does not halt ($t_\tau$-central split) is determined by the counter probability of the case that DPM halts.
\begin{align*}
     \Pr[s\in W_{> t_{\tau}}] &= 1 - \Pr[\text{DPM halts immediately on }S]\\ 
     & = 1-\Pr[s\in W_{\le t_{\tau}}]\\
     & = 1-\frac{\sum_{s\in W_{\le t_{\tau}}} \exp(f(S,\tilde n, s)\varepsilon/(2\Delta_f))}{\sum_{s\in W} \exp(f(S,\tilde n, s)\varepsilon/(2\Delta_f))}
\end{align*}
To get a lower bound on the probability that DPM does not halt, we give an upper bound on the probability that DPM halts. First we upper bound the nominator $\sum_{s\in W_{\ge t_\tau}} \exp(f(S,\tilde n, s)\varepsilon/(2\Delta_f))$. 
The centreness of split candidates in $W_{\le t_\tau}$ can be up to $t_\tau$ and the emptiness up to $1$ with the assumption of $ t_\tau < t$. Then, we get $\sum_{s\in W_{\le t_{\tau}}} \exp(f(S,\tilde n, s)\varepsilon/(2\Delta_f)) \le |W_{\le t}|\exp((\alpha + t_{\tau})\varepsilon/(2\Delta_f))$.
To get a lower bound on the denominator, i.e. the scores of all split candidates $\sum_{s\in W}\exp{(f(S,\tilde n, s)\varepsilon/(2\Delta_f)})$,we distinguish the following sets of split candidates: (1) split candidates that cause DPM to halt with lower bound on the score of $\emin$. (2) The splits with centreness larger than $t_\tau$ but less than $t$ and again minimum emptiness of $\emin$. (3) We are left with the split candidates in the inner quantile with minimum emptiness of $\emin$ and centreness of $t$. Finally, we can put this together to lower bound the probability that DPM does not halt as follows:
\begin{align*}
    &\Pr[s\in W_{> t_{\tau}}]\\
    &\ge 1-\frac{|W_{\le t_{\tau}}| e^{(\alpha + t_\tau)\varepsilon/(2\Delta_f)}}{|W_{\le t_{\tau}}|e^{\emin\alpha \varepsilon/(2\Delta_f)}+ |W_{> t_{\tau}, \le t}|e^{(\emin\alpha + t_\tau)\varepsilon/(2\Delta_f)} + |W_{\ge t, >t_\tau}|e^{(\emin\alpha + t)\varepsilon/(2\Delta_f)}}
 \end{align*}
\end{proof}

\begin{remark}
    Note that if $t_\tau > t$, $|W_{> t_{\tau}, \le t}| = 0$ and the set of split candidates can be divided into split candidates that violate the minimum cluster size $W_{\le t_{\tau}}$ and all other split candidates $W_{> t_{\tau},\ge t}$.
\end{remark}

As in \cref{Sssec:tPrimeCentral}, we distinguish the different cases $t_\tau < t$ and $t_\tau \ge t$. For a split such that DPM does not halt, the minimum centreness threshold improves the least for the largest possible set. Thus, for both cases, we can assume $|S'| = \tilde n_S - \tau_e > \tau_e$ where $\tau_e$ is the minimum cluster size. Again, we set $t=2q$ to get a lower bound on the centreness function $c_{2q,q} = 1 - \big|\frac{2\tau_e}{\tilde n_S}-1\big|$ and as we know $2\tau_e \le \tilde n_S$ because otherwise DPM would have already halted. Then, we know $t_\tau(S)= \frac{2\tau_e}{\tilde n_S}$. To analyse the improvement of the minimum centreness threshold, it suffices to adjust the $\tilde n_S$ accordingly, i.e. $\tilde n_S = \tilde n_S - \tau_e$.
\begin{align*} 
    &t_\tau (S') \ge \frac{2\tau_e}{\tilde n_S -\tau_e} = \underbrace{\frac{2\tau_e}{\tilde n_S}}_{t_\tau(S)} \cdot  \underbrace{\frac{\tilde n_S}{\tilde n_S -\tau_e}}_{>1} > t_\tau(S)
\end{align*}

The approximation of $t_\tau$ on split level $i$ with $S$ as the origin set can be achieved as follows. With the exception of the assumption of $t'$-central splits, the centreness threshold is dependent solely on the minimum cluster size $\tau_e$ and the size of the origin set $S$.
\begin{lemma}
\label{lem:minTtauGen}
Given a set $S$ and its noisy count $\tilde n$ and DPM as introduced in \cref{sec:dpm} with centreness parameters $t,q$. Assuming that there is at least one $t'$-central split, the minimum centreness threshold for any subset of $S$ after $i$ spits is
    \begin{align}
        t^i_\tau(S) &\ge \frac{2\tau_e}{\tilde n_S} \cdot \bigg(\frac{\tilde n_S}{\tilde n_S -\tau_e}\bigg)^i\text{.}\label{eq:centThresholdGen}
    \end{align}
\end{lemma}

The aforementioned building blocks allow us to establish a lower bound on the probability that DPM halts after at most $i$ splits. It should be noted that we are considering the case in which there is no knowledge of the existence of a $t'$-central split.
\begin{theorem}[Lower bound on probability that DPM halts after $i$ splits]
Given a set $S$ and its noisy count $\tilde n$, and corresponding centreness threshold $t_\tau(S)$ as well as a set of split candidates $W$. Further, given DPM as introduced in \cref{sec:dpm} with scoring function $f$ with sensitivity $\Delta_f$, metric weight $\alpha$, centreness parameters $t,q$ and privacy budget $\varepsilon$ for the Exponential Mechanism. Assuming that the minimum occurring emptiness is $e_{\min}$ and the maximum emptiness in the inner quantile is $e_{Q_I}$,
Then, for a fixed $t'$ the probability that DPM halts until recursion level $i$ can be lower bounded as follows.
   \begin{align}
   \label{eq:finalProbHaltGen_untight}
        &\Pr[\text{DPM halts after }i \text{ splits on }S] \\\ge&
        \sum^j_{i = 0} \underbrace{\Pr[s\in W_{\le t_\tau^i(S)}]}_{\cref{lem:lowBoundHaltImm}} \cdot \bigg(\prod^j_{\ell = 0} \underbrace{\Pr[s\in W_{\ge t^\ell_\tau(S)}]}_{\cref{lem:lowerBoundNotHalt_untight}}\bigg) ^{2^i}\\
         = &\sum^j_{i = 0} \frac{|W_{\le t^j_\tau(S)}| e^{\emin\varepsilon/(2\Delta_f)}}{|W_{\le t_{\tau(S)}}|e^{(t^j_{\tau}(S)+\alpha)\varepsilon/(2\Delta_f)}+ |W_{\ge t}|e^{(1 +\alpha e_{Q_I})\varepsilon/(2\Delta_f)}+ |W_{>t^j_{\tau}(S)\land <t}|e^{(t+\alpha)\varepsilon/(2\Delta_f)}} \\
         & \mkern-6mu\bigg(\mkern-2mu\prod^j_{\ell = 0} 1\mkern-2mu-\mkern-2mu\frac{|W_{\le t_{\tau}}| e^{(\alpha + t_\tau)\varepsilon/(2\Delta_f)}}{|W_{\le t_{\tau}}|e^{\emin\alpha \varepsilon/(2\Delta_f)}\mkern-2mu+\mkern-2mu |W_{> t_{\tau}, \le t}|e^{(\emin\alpha + t_\tau)\varepsilon/(2\Delta_f)}\mkern-2mu +\mkern-2mu |W_{\ge t, > t_\tau}|e^{(\emin\alpha + t)\varepsilon/(2\Delta_f)}} \mkern-4mu \bigg) ^{2^i}\\
        &\text{ with }t^j_\tau(S) \ge \frac{2\tau_e}{\tilde n_S} \cdot \bigg(\frac{\tilde n_S}{\tilde n_S -\tau_e}\bigg)^i
    \end{align} 
\end{theorem}
\begin{proof}
    In order to approximate a lower bound on the probability that DPM halts immediately for a given set $S$ and a recursion level $i$, we take the centreness threshold $t^j_ \tau(S)$ as in \cref{eq:centThresholdGen} and apply \cref{lem:lowBoundHaltImm}. If DPM does not halt immediately, the probability that DPM halts on the subsequent recursion levels can be determined by multiplying the aforementioned probability by the condition that such a split is selected (\cref{lem:lowerBoundCentral_untight}).  In order to guarantee that DPM halts on level $i$, there can be $2^i$ subsets, and DPM must halt for all of them.
\end{proof}

The provided bounds on the probability that DPM halts appropriately depends on the data set in question, particularly the number of split candidates that cause DPM to halt. This number is in turn contingent upon the minimum size of a cluster, as well as the input data itself. How are the data points distributed and thus how does the number of split candidates change that violate the minimum cluster size for an increasing centreness threshold $t_\tau$.

\subsection{Limitations of Stopping Criterion}
In some cases, DPM does not halt appropriately. This phenomenon can be explained by considering the characteristics of the data sets on which DPM does not halt appropriately. It is important to note that for DPM to halt on a given set, it must halt on each subset.

\subsubsection{Equally Distributed Data Set}
Let us consider a data set in which the data points are distributed equally across the entire range for each dimension. In such a data set, we can assert with certainty that there are no clusters and no suitable split candidates. Consequently, DPM should terminate within a few steps and not return clusters.

The interval sizes for all splits are identical, as are the emptiness values of all split candidates, given that the data points are distributed equally. Consequently, the selection of split candidates is based exclusively on their centreness values. For the sake of simplicity, we shall assume that the optimal score is always selected, which results in the set being split into two subsets of equal size. We know that after $i$ splits on a set $S$, the remaining subsets are of size $\frac{|S|}{2^i}$. If the minimum cluster size $\tau_e$ is less than this size, i.e. $\tau_e < \frac{|S|}{2^{\tau_s}}$ where $\tau_s$ is the maximum recursion depth, DPM does not halt appropriately.  

\subsubsection{Gaussian Distributed Data Set}
We proceed to examine data points derived from a Gaussian distribution per dimension. In these data sets, we posit that there are no split candidates with high emptiness values in proximity to the centre, indicating that a single cluster has been identified that does not necessitate further division. 
In the event that the impact of emptiness is greater than that of centreness, the probability of a split being selected that contravenes the minimum cluster size increases. We discuss the circumstances under which, even in the optimal case, DPM will not halt appropriately. Consequently, we consider the optimal case in which, if DPM does not select a split candidate that immediately halts, it selects the split with the best centreness value (of $1$).
In the case of a single Gaussian distribution per dimension, the split interval size estimation implemented by DPM yields a value of $1/2\sigma$. 
Subsequently, for each value of $i$ (up to $7$ as $\tau_s =7$ in the experiments), we provide the minimum emptiness of the most central split (in all subsets) as well as the maximum score of split candidates that cause to halt, which we will refer to as the halting split candidates. 

% Score of halting split candidates
In selecting the optimal split candidate with the objective of halting DPM, we assume that $S = S'\dot \cup S''$ with $|S'|=|S''|$. Therefore, it can be demonstrated that the centreness of splits that cause DPM to halt is at most the centreness value belonging to a split candidate with rank $\tau_e$ which we will denote as the centreness threshold $t_\tau$. Although the minimum cluster size remains constant, the centreness threshold increases, doubling in our case. The emptiness of the halting split candidates is at most one, allowing us to derive an upper bound on the score of a halting split candidate of $s^h_i = 2^i \cdot t_\tau + \alpha$ for the $i$-th recursion level.

% Score of conflicting split candidate
In this analysis, the score of the halting split candidates is compared with that of the most central split candidate, which has a centreness value of $1$.  It is then necessary to provide a lower bound on the emptiness value for different values of $i$. Recall that the emptiness is determined as the difference between the optimal emptiness of $1$ and the fraction of elements in the current set.  In the case of data points drawn from a single Gaussian distribution, , the emptiness can be determined by assuming that the most central split is always selected. The shift from the mean of the Gaussian to the median of the subset after $i$ splits, denoted as $z_i$, can be determined by the range that holds $1-2^{-i}$ of the data points. The $z_i$ are computed as follows. For $i< \tau_s$, the precise value of the shift is provided.
\begin{align*}
    z_i &= \Phi^{-1}\bigg(\frac{1 + (1-2^i)}{2}\bigg)\\
    z_0 &= \Phi^{-1}\bigg(\frac{1}{2}\bigg) = 0\\
    z_1 &= \Phi^{-1}\bigg(\frac{1.5}{2}\bigg) = \Phi^{-1}(0.75) \approx 0.6744\\
    z_2 &= \Phi^{-1}\bigg(\frac{1.75}{2}\bigg) = \Phi^{-1}(0.875) \approx 1.15\\
    z_3 &= \Phi^{-1}\bigg(\frac{1.875}{2}\bigg) = \Phi^{-1}(0.9375) \approx 1.53\\
    z_4 &= \Phi^{-1}\bigg(\frac{1.9375}{2}\bigg) = \Phi^{-1}(0.9688) \approx 1.86\\
    z_5 &= \Phi^{-1}\bigg(\frac{1.9688}{2}\bigg) = \Phi^{-1}(0.9844) \approx 2.13\\
    z_6 &= \Phi^{-1}\bigg(\frac{1.9844}{2}\bigg) = \Phi^{-1}(0.9922) \approx 2.41
\intertext{
The split interval size is fixed as $\beta=1/2\sigma$, and thus, it is necessary to consider the fraction of elements in the range $[-z_i-\beta/(2\sigma), -z_i+\beta/(2\sigma)]$ or $[z_i-\beta/(2\sigma), z_i+\beta/(2\sigma)]$.  It is also necessary to consider that the fraction of elements obtained from the normal distribution is related to the input data set rather than the current subset. As the partitions are of the same size, this introduces a factor of $2^i$. Consequently, for a $z_i$, the lower bound on the emptiness of a central split can be expressed as follows, where $S_i$ denotes the set on recursion level $i$, and $s^c_i$ the corresponding central split.}
    e^c_i &= e(S_i, |S|\cdot 2^{-i} ,s^c_i) \ge 1-\frac{|s^c_i|}{|S|}\cdot 2^i \\
    &= 1 - 2^i \cdot \frac{1}{2} \bigg( 2\Phi\bigg(z_i + \frac{\beta}{2\sigma}\bigg)- 1 -\bigg( 2\Phi\bigg(z_i - \frac{\beta}{2\sigma}\bigg)\bigg) -1 \bigg) \\
    &= 1 - 2^i\bigg(\Phi\bigg(z_i + \frac{\beta}{2\sigma}\bigg) - \Phi\bigg(z_i - \frac{\beta}{2\sigma}\bigg)\bigg)
\intertext{The next step is to plug in the different $z_i$ and set $\beta = 1/2\sigma$ in order to obtain the minimum emptiness for each split level.}\\
    e^c_0 &=  1 - \bigg(\Phi\bigg(\frac{1}{4}\bigg) - \Phi\bigg(- \frac{1}{4}\bigg)\bigg) = 0.80258 \\
    e_1 &= 1 - 2^1\bigg(\Phi\bigg(z_1 + \frac{1}{4}\bigg) - \Phi\bigg(z_1 - \frac{1}{4}\bigg)\bigg) \\
    &= 1-2 \big(\Phi(0.9244) - \Phi(0.4244)\big) = 1-2(0.82121 - 0.66276) \\&= 0.6831 \\
    e^c_2 &= 1 - 2^2\bigg(\Phi\bigg(z_2 + \frac{1}{4}\bigg) - \Phi\bigg(z_2 - \frac{1}{4}\bigg)\bigg) \\
    &= 1-4 \big(\Phi(1.4) - \Phi(0.9)\big) = 1-4(0.91924 - 0.81594) \\&= 0.5868 \\
    e^c_3 &= 1 - 2^3\bigg(\Phi\bigg(z_3 + \frac{1}{4}\bigg) - \Phi\bigg(z_3 - \frac{1}{4}\bigg)\bigg) \\
    &= 1-8 \big(\Phi(1.78) - \Phi(1.28)\bigg) = 1-8(0.96246 - 0.89973) \\&= 0.49816\\
    e^c_4 &= 1 - 2^4\bigg(\Phi\bigg(z_4 + \frac{1}{4}\bigg) - \Phi\bigg(z_4 - \frac{1}{4}\bigg)\bigg) \\
    &= 1-16 \big(\Phi(2.11) - \Phi(1.61)\big) = 1-16(0.98257 - 0.94630)\\& = 0.41968\\
    e^c_5 &= 1 - 2^5\bigg(\Phi\bigg(z_5 + \frac{1}{4}\bigg) - \Phi\bigg(z_5 - \frac{1}{4}\bigg)\bigg) \\
    &= 1-32 \big(\Phi(2.38) - \Phi(1.88)\big) = 1-32(0.99134 - 0.96995) \\&= 0.3155 \\
    e^c_6 &= 1 - 2^6\bigg(\Phi\bigg(z_6 + \frac{1}{4}\bigg) - \Phi\bigg(z_6 - \frac{1}{4}\bigg)\bigg) \\
    &= 1-64 \big(\Phi(2.66) - \Phi(2.16)\big) = 1-64(0.99609 - 0.98461) \\&= 0.26528 \\
\end{align*}
% Compare halting split candidate and and central split candidate
At last, we are in a position to make a comparison between the highest score achieved by a halting split candidate and the lower bound on the score of a central split candidate. It is only in cases where the halting split candidate's score exceeds that of the central split by a certain margin that it is selected with a high probability.
\begin{align}
    s^h_i &> s^c_i + m \nonumber\\
    2^i \cdot t_\tau + \alpha &> 1+e^c_i\alpha + m \nonumber\\
    t_\tau &> \frac{1+(e^c_i -1)\alpha + m}{2^i} \label{ineq:centrenessThreshold}
\end{align}
In order to determine the corresponding centreness threshold that ensures DPM halts on a given recursion level, we proceed to plug in the values for all levels up to and including the sixth. It is established that $t_\tau $ is within the range of $[0,1]$. A high value of $t_\tau$ indicates that DPM will only halt when the minimum cluster size is considerable. \cref{fig:centrenessThreshold} depicts the lower bound of inequality in \cref{ineq:centrenessThreshold} for $m=0$. In order to ensure that DPM halts on all subsets, it is essential to have $m$ is significantly greater than zero. As the emptiness weight $\alpha$ increases, the lower bound on $t_\tau$ decreases. For values of $\alpha$ less than $5$, the lower bound on $t_\tau$ is too large for DPM to halt until a later recursion level (for $i\ge3$). For larger values of $\alpha$, namely $\alpha \ge 5$, the lower bound on $t_\tau$ is less than zero. This allows for a larger value of $m$, thereby ensuring the efficacy of the splitting process. This is due to the fact that when the emptiness has a greater impact than the centreness, there is a high probability that DPM will select a split with a low centreness value but a high emptiness value.  It is probable that these splits will be in close proximity, given that in the considered data set there are no gaps present and that it is Gaussian distributed. Furthermore, that if $\alpha$ is excessively high, the utility will be adversely affected, as discussed in \cite{DPM}.
For all values of $\alpha$, the lower bound on the centreness threshold for $t_\tau$ converges to zero as the number of recursion levels increases. This implies that even for small minimum cluster sizes, the scores of splits that cause DPM to halt are larger than those of central splits. Thus, even with a margin of $m>0$, a reasonable $t_\tau$ is obtained. Consequently, even if the maximum recursion level is set too high for a given data set, DPM will halt as a result of the minimum cluster size criterion.

\begin{figure}[t!]
    \centering
    \captionsetup{format=plain, indention=0cm}
    \includegraphics[width=\linewidth]{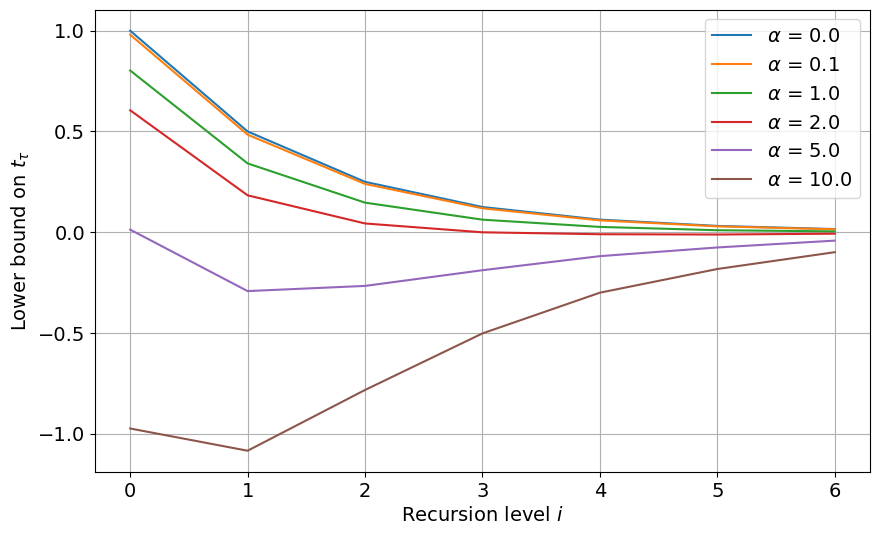}
    \caption{Visualisation of the lower bounds on the centreness threshold as given in \cref{ineq:centrenessThreshold} for different values of the parameter $\alpha$ with $m=0$.he x-axis represents the recursion level, while the y-axis depicts the corresponding centreness threshold. Despite the fact that $t_\tau$ is greater than or equal to zero, negative $y$ values are possible, enabling the use of large m while maintaining the high probability of DPM halting. Larger values for $t_\tau$ result in DPM only halting for larger minimum cluster size.}
    \label{fig:centrenessThreshold}
\end{figure}
\subsection{Discussion}
A theoretical analysis of algorithms may identify potential issues or areas of vulnerability within the algorithmic structure. It was  demonstrated that the guarantees presented in \cref{eq:finalProbHaltTprime_untight} are, in fact, quite loose. This is primarily due to approximating the probability that a $t'$-central split is selected by the probability of one single $t'$-central split.  In the absence of the assumption of at least one $t'$-central split, the lower bounds presented in \cref{eq:finalProbHaltGen_untight} are even less tight. As we considered the worst-case scenarios, our estimates are likely to be conservative in the majority of cases where there may be multiple $t'$-central candidates. This issue is particularly prevalent for smaller values of $t'$. 

The underlying issue is that even when the remaining data sets span only a small portion of the original range, the split candidates are not recalculated. Consequently, depending on the data set, there may be only one or two split intervals that contain data points. While this appears to be a crucial aspect for enhancing the utility of DPM, it also necessitates updating the split interval size, which in turn may compromise the confidentiality of the data.

%% file: main/silhScore.tex
\section{Theoretical Bounds Regarding the Silhouette Score}
\label{cha:silhScore}
In \cite{DPM}, the authors demonstrated that DPM has a high probability of selecting optimal splits. Additionally, the experiments revealed that the selected splits yielded excellent clustering quality, as evidenced by the standard metrics of inertia, accuracy, and the silhouette score.
In order to gain further insight into the theoretical utility guarantees of DPM with regard to the standard metrics, we establish a link between the selection of optimal splits and the recognised clustering metric, namely the silhouette score. We also discuss the shortcomings of the silhouette score as a clustering metric by giving an example where the silhouette score decreases for an optimal DPM-based split.

\subsection{Silhouette Score}
\label{sec:silhScore}
The silhouette score of a clustering result is determined by a comparison of the inter- and intra-cluster distances. Consequently, for each data point, the distance between that point and its assigned cluster centre is calculated (denoted as $\dist(C_i,x)$), as well as the distance between that point and the cluster centre of the nearest other cluster (denoted as $\dist(C_j,x)$). Should the assigned cluster centre prove to be closer than any other cluster centre, the silhouette value will be positive; conversely, it will be negative should this not be the case.  The silhouette values are subsequently normalised by dividing them by the maximum value of the inter- and intra-distance. Thus, the optimal silhouette score of $1$ is achieved when all data points are positioned at their respective cluster centres. The silhouette score of a clustering result is then calculated as the mean silhouette value across all data points and clusters. 
Consequently, the clustering result with the best silhouette score is that in which each data point constitutes a single cluster. Therefore, in the case that a cluster only holds one data point, its silhouette value is set to 0.

\begin{definition}[Silhouette Score]
    \label{def:silhscore}
    Given a data set $D$ and a set of $k$ cluster centres $C = \{c_0, \dots, c_{k-1}\}$, then the silhouette score $SC$ of the clustering $C$ of data set $D$ is computed as follows: 
    
    \begin{align*}
        S(C,D) &= \frac{1}{n} \sum_{i=0}^k\sum_{x \in C_i} s(C_i,x)\\\text{ with } s(C_i,x) &= \begin{cases}
            \frac{\dist(C_j,x) - (C_i,x)}{\max(\dist(C_j,x),(C_i,x))}, & |C_i|>1\\
            0, & \text{else} 
        \end{cases} \\
        \text{and } \dist(C_i,x) &= \frac{1}{|C_i| - 1} \sum_{x' \in C_i / \{x\}} || x' - x ||_2\text{, } 
        \dist(C_j,x) = \min_{j \neq i} \left( \frac{1}{|C_j|} \sum_{x' \in C_j} || x' - x ||_2 \right) \text{.}
    \end{align*}
\end{definition}

\subsection{DPM and Silhouette Score}

% How does the silhouette score change, when a split is selected
The manner in which the silhouette value of data points may undergo alteration is contingent upon a number of factors. The application of a DPM-based split results in a modification of the clustering outcome for a single subset. Consequently, our analysis focuses on the impact of these splits on the clustering result. The total silhouette score is the average of the silhouette values of all data points. Consequently, we initially examine the alteration in the silhouette value of a single data point subsequent to splitting. With a theoretical delineation of the change in the silhouette value for a DPM-based split, we then proceed to discuss the limitations of this characterisation. It is still possible for a suitable split according to DPM to result in a reduction in the overall silhouette score.

\subsection{Change in Silhouette Value}
In examining the alteration of the silhouette value for a specific data point subsequent to a DPM-based split,  it is essential to differentiate between two scenarios: the data point in question is within the subset undergoing division, or it is a member of one of the other subsets. In the former case, both the intra- and inter-distances may undergo modification. Conversely, in the latter scenario, only the inter-distance is influenced by the split.

\paragraph{Partition that is Split}
We initially examine the potential alterations to the silhouette value for data points within the subset undergoing division.
Subsequently, following the partitioning, each data point that was in $S_0$ is situated within one of the resulting partitions, designated as either $S_0'$ or $S_0''$. 
To determine the change in the silhouette value, we need to distinguish the possible changes in the intra- and inter-distance for both $S_0'$ and $S_0''$. The intra-distance for data points in the partitions are now the distances to their assigned cluster centre. For the inter-distance, either the closest cluster to $S_0$ remains the closest cluster or the other partition. In both cases the intra-distance can be larger or smaller than the inter-distance.

For the silhouette value of data points in the subset to be split, both the inter- and intra-distances can be decreased or increased. Therefore, if the intra-distance is reduced more than the inter-distance, or the inter-distance is increased more than the intra-distance, the silhouette value of these data points will improve after splitting.

\begin{lemma}[Improvement of data points in partitioned subset]
    For each data point $x$ in a given set $S_0$, the nearest neighbouring cluster is indicated by $C$. It is assumed that $S_0$ is being divided into $S_0'$ and $S_0''$ via a DPM-based split.
    Subsequently, if the data point $x$ is situated in the subset $S_0'$, the silhouette value for $x$ will improve following the split in the following cases: $\dist(C,x) < \dist(S_0'',x)<\dist(S_0,x)$; or $\dist(C,x)\ge\dist(S_0'',x)$ and $\dist(S_0'',x) - \dist(S_0',x) > \dist(C,x) - \dist(S_0,x)$.
    It should be noted that the change in the silhouette value is analogous for the case of $x \in S_0''$.
\end{lemma}
\begin{proof}
The silhouette value for data point $x$ in set $S_0$ is given by the following equation:
\begin{align*}
    s(S_0,x) &= \frac{\dist(C,x) - \dist(S_0,x)}{\max(\dist(C,x),\dist(S_0,x))}
    \intertext{The silhouette value of $x$ after a DPM-based split is either $s(S_0',x)$ or $ s(S_0'',x)$ depending on the assignment of $x$ after the split. Then by \cref{def:silhscore}, the silhouette value after splitting is as follows. We give the silhouette value for both cases $x\in S_0'$ and $x\in S_0''$ but note that for each data point that was in $S_0$ we compute either $s(S_0',x)$ or $s(S_0'',x)$.}
    s(S_0',x) &= \frac{\min(\dist(C,x), \dist(S_0'',x)) - \dist(S_0',x)}{\max(\min(\dist(C,x), \dist(S_0'',x)),\dist(S_0',x))} \text{ or}\\
    s(S_0'',x) &= \frac{\min(\dist(C,x), \dist(S_0',x)) - \dist(S_0'',x)}{\max(\min(\dist(C,x), \dist(S_0',x)),\dist(S_0'',x))}
\end{align*}
The inter- and intra-distance of $x$ may change. To analyse the change in the silhouette value, we distinguish the following cases. The intra-distance $s(S_0,x)$ is replaced by either $s(S_0',x)$ or $s(S_0'',x)$ depending on the assignment of $x$.
    \begin{enumerate}
        \item Cluster $C$ remains the closest cluster to $x$. If $x\in S_0'$, this implies $\dist(C,x)< \dist(S_0'',x)$ and if $x\in S_0''$, this implies $\dist(C,x)< \dist(S_0',x)$. In both cases, the silhouette value improves compared to $s(S_0,x)$ if the intra-distance decreases, i.e. $x$ is closer to the cluster centre than before the split.
        \begin{align*}
            s(S_0',x) &= \frac{\dist(C,x)- \dist(S_0',x)}{\max(\dist(C,x),\dist(S_0',x))} \begin{cases}
                > s(S_0,x)   &  \text{if }\dist(S_0',x) < \dist(S_0,x) \\
                \le s(S_0,x)   &  \text{if }\dist(S_0',x) \ge \dist(S_0,x) 
            \end{cases}\\
            s(S_0'',x) &= \frac{\dist(C,x) - \dist(S_0'',x)}{\max(\dist(S_0',x),\dist(S_0'',x))} \begin{cases}
                > s(S_0,x)   &  \text{if }\dist(S_0'',x) < \dist(S_0,x) \\
                \le s(S_0,x)   &  \text{if }\dist(S_0'',x) \ge \dist(S_0,x) 
            \end{cases}\\
        \end{align*}
        \item Cluster $C$ is not the closest cluster to $x$ after the split but the other partition. If $x\in S_0'$, this implies that $\dist(C,x)\ge \dist(S_0'',x)$ and if $x\in S_0''$, this implies $\dist(C,x)\ge \dist(S_0',x)$. The silhouette value improves compared to $s(S_0,x)$ if the difference between the inter- and intra-distances increases. 
        \begin{align*}
            s(S_0',x) &= \frac{\dist(S_0'',x)- \dist(S_0',x)}{\max(\dist(S_0'',x),\dist(S_0',x))}\\
            &\mkern-16mu \begin{cases}
                > s(S_0,x)   &  \text{if }\dist(S_0'',x) - \dist(S_0',x) > \dist(C,x) - \dist(S_0,x) \\
                \le s(S_0,x)   &  \text{if }\dist(S_0'',x) - \dist(S_0',x) \le \dist(C,x) - \dist(S_0,x)
            \end{cases}\\
            s(S_0'',x) &= \frac{\dist(S_0',x) - \dist(S_0'',x)}{\max(\dist(S_0',x),\dist(S_0'',x))}\\
            &\mkern-16mu \begin{cases}
                > s(S_0,x)   &  \text{if }\dist(S_0',x) - \dist(S_0'',x) > \dist(C,x) - \dist(S_0,x) \\
                \le s(S_0,x)   &  \text{if }\dist(S_0',x) - \dist(S_0'',x) \le \dist(C,x) - \dist(S_0,x)
            \end{cases}\\
        \end{align*}
    \end{enumerate}
\end{proof}

\paragraph{All other partitions}
Since we know how the silhouette value changes for the elements in the subset being split ($S_0$), we need to analyse the effect on all data points that are not in the subset being considered ($S_j$). For these data points, we know that the intra-cluster distance does not change, so we only need to consider all cases for the change in inter-distance.  So we only consider the silhouette value $sc = s(S_j,x)$ and after splitting $sc' = s(S_j,x)$ for all cases.  After splitting, the nearest cluster centre for a data point in $S_j$ can change as follows.
The closest cluster centre for a data point is $S_0$ and after $S_0$ is split, either $S_0'$, $S_0''$ or some $C$ can be the closest cluster. If before the split the nearest cluster centre is some cluster $C$; after the split either $C$ remains the nearest cluster, and thus the inter-distance remains or $S_0'$ or $S_0''$ are the closest cluster centre. 
Only the assumption that $\dist(C,x)>\dist(S_0,x)$ directly distinguishes two cases for $sc$ and $sc'$. All other assumptions to distinguish all possible $sc$ do not directly imply different cases for $sc$. 

For data points in subsets that are not partitioned, the silhouette value can only improve if the inter-distance increases while the intra-distance remains the same. Only if the partitioned set was the closest cluster before, the inter-distance can increase if the distance between $x$ and both resulting partitions $S_0'$ and $S_0''$ is greater than between $x$ and $S_0$ before.

\begin{lemma}[Change in the silhouette value for data points not in partitioned subset]
    For each data point $x$ in a given set $S_j$, the nearest neighbouring cluster is either some cluster $C$ or $S_0$. It is assumed that $S_0$ is being divided into $S_0'$ and $S_0''$ via a DPM-based split and further that $S_0'$ is closer to $x$ than $S_0''$.
    The silhouette value for $x$ will improve following the split in the case: $\dist(S_0,x) < \min(\dist(S_0',x), \dist(S_0'',x))$ (1a, 1b).
\end{lemma}
\begin{proof}
The silhouette value before the split is defined as $sc=s(S_j,x)$ for some data point $x\in S_j$ and after a DPM-based split $sc'=s(S_j,x)$.
    \begin{align*}
    sc = s(S_j,x) &= \frac{\min(\dist(C,x),\dist(S_0,x))-\dist(S_j,x)}{\max(\min(\dist(C,x),\dist(S_0,x)),\dist(S_j,x))}\\
    sc' = s(S_j,x) &= \frac{\min(\dist(C,x),\dist(S_0',x), \dist(S_0'',x))-\dist(S_j,x)}{\max(\min(\dist(C,x),\dist(S_0',x), \dist(S_0'',x)),\dist(S_j,x))}
\end{align*}
To determine the change in the silhouette value of $x$, the following cases of change in the inter- and intra- distance of $x$.
\begin{enumerate}
    \item We first consider the case where the cluster $S_0$ that is split was the closest cluster for a data point $x$. After the split, there is no cluster $S_0$ and either one of the partitions $S_0', S_0''$ is the closest cluster or if the distance to both centres has increased, some cluster $C$ may also be the closest cluster centre.
    \[sc = s(S_j,x) = \frac{\dist(S_0,x)-\dist(S_j,x)}{\max(\dist(S_0,x),\dist(S_j,x))}\]
    \begin{enumerate}
        \item After the split, instead of $S_0$, one of its partitions $S_0',S_0''$ is the nearest cluster. In both cases, $S_0'$ or $S_0''$ as the nearest neighbouring cluster, the silhouette value for a data point $x$ improves if the inter-distance is greater than before the split. The intra-distance is not affected by the split.
        \begin{align*}Damm in
            sc' = s(S_j,x) &= \frac{\min(\dist(S_0',x), \dist(S_0'',x))-\dist(S_j,x)}{\max(\dist(S_0',x),\dist(S_j,x))}\\
            &\mkern-16mu \begin{cases}
                > sc   &  \text{if } \min(\dist(S_0',x), \dist(S_0'',x))> \dist(S_0,x) \\
                \le sc   &  \text{if }\min(\dist(S_0',x), \dist(S_0'',x)) \le \dist(S_0,x)
                 \end{cases}
        \end{align*}
        \item After the split, neither $S_0'$ nor $S_0''$ is the closest cluster for a data point $x$, but some cluster $C$. In this case, we know that the distance to $C$ is greater than the inter-distance to $S_0$. So, the silhouette value will always improve.
        \[sc' = s(S_j,x) = \frac{\dist(C,x)-\dist(S_j,x)}{\max(\dist(C,x),\dist(S_j,x))} > sc\]
    \end{enumerate}
    \item If $S_0$ was not the closest cluster before the split, the silhouette value will only be affected if one of the partitions is the closest cluster after the split.
    \[sc = s(S_j,x) = \frac{\dist(C,x)-\dist(S_j,x)}{\max(\dist(C,x),\dist(S_j,x))}\]
    \begin{enumerate}
        \item In both cases, $S_0'$ or $S_0''$ as the closest cluster, the the inter-distance for $x$ is smaller than before the split, otherwise $C$ would still be the closest cluster. The intra-distance is not affected by the split, so there is no case where the silhouette value improves.
        \begin{align*}
            sc' = s(S_j,x) &= \frac{\min(\dist(S_0',x), \dist(S_0'',x))-\dist(S_j,x)}{\max(\min(\dist(S_0',x), \dist(S_0'',x)),\dist(S_j,x))} < sc 
        \end{align*}
        \item If $C$ is still the closest cluster after the split, the inter-distance is not affected. Since the intra-distance is also unaffected, the silhouette value remains the same.
        \[sc' = sc\]
    \end{enumerate}
\end{enumerate}
\end{proof}

The silhouette score of a clustering is determined by averaging the silhouette value per data point. Thus, the silhouette score improves if the silhouette value per data point improves on average. For all data points in clusters that are not split, the silhouette value will only improve if $S_0$ was the closest cluster before the split and after the split, the centres of both partitions are further away from $x$. For the data points in the cluster that is split, the improvement of the silhouette value depends on the change of the intra-distance compared to the inter-distance.

\subsection{Flaws of Silhouette Score}
In the previous section we discussed in which cases the silhouette value for a data point improves. If, on average, the improvements are greater than the degradations, then the silhouette score for a data set $SC$ improves. In this section, we analyse example data sets for which the silhouette score does not improve after applying a DPM-based split. 

We consider a data set with three clusters, each drawn from an isotropic Gaussian distribution with the same standard deviation. Since we are considering the change in the silhouette score after a DPM-based split, we assume the following setting for the data set. We assume that all clusters have the same number of data points. An example data set is plotted in \cref{fig:CounterEx_silhScore_before}. The initial setting is that DPM has already chosen a split that separates a third of the data points from the rest, resulting in the clustering as shown in the figure. We refer to the purple data points as cluster $C$ and to the yellow data points as $S_0$. Analogous to the previous section, we now consider the case where the next DPM-based split, divides $S_0$ into $S_0'$ (cyan in \cref{fig:CounterEx_silhScore_after}) and $S_0''$ (yellow in \cref{fig:CounterEx_silhScore_after}). For the example data set, the silhouette score before the second split decreases after the split. 
% Counterexample: When is silhouette score not good?
% \begin{figure}[t!]
% \centering
% \captionsetup{format=plain, indention=0cm}
% \subfloat[Clustering result before splitting with $k=2$ and the corresponding silhouette score of $0.72$. \label{fig:CounterEx_silhScore_before}]{\includegraphics[width=.45\linewidth]{res/figures/CounterEx_silhScore_before.png}}
% \subfloat[Clustering result before splitting with $k=3$. and the corresponding silhouette score of $0.7$. \label{fig:CounterEx_silhScore_after}]
%   {\includegraphics[width=.45\linewidth]{res/figures/CounterEx_silhScore_after.png}}

\begin{figure}[t!]
    \centering
    \begin{subfigure}{0.45\textwidth}
    \centering
    \includegraphics[width=\linewidth]{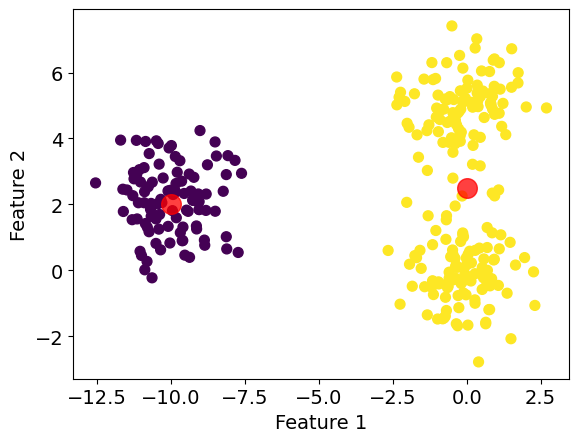}
    \caption{Clustering result before splitting with $k=2$ and the corresponding silhouette score of $0.72$.}
    \label{fig:CounterEx_silhScore_before}
    \end{subfigure}
    ~\hspace{2em}
    \begin{subfigure}{0.45\textwidth}
    \centering
    \includegraphics[width=\linewidth]{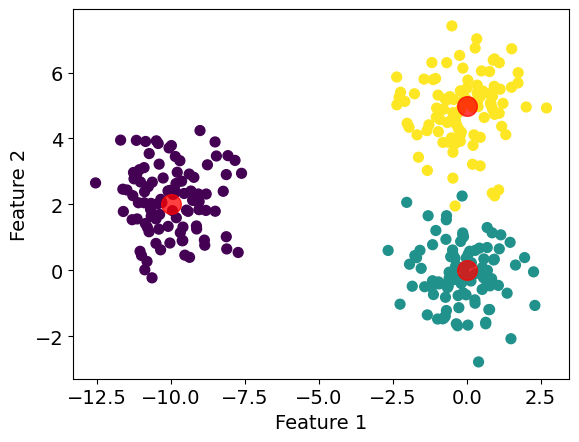}
    \caption{Clustering result before splitting with $k=3$. and the corresponding silhouette score of $0.7$.}
    \label{fig:CounterEx_silhScore_after}
    \end{subfigure}
\caption{Example data set of three clusters each drawn from an isotropic Gaussian distribution with a standard deviation of $\sigma =1$. The two clusters with centres $(0,5)$ and $(0,0)$ are $5\sigma$ apart. In \cref{fig:CounterEx_silhScore_before} the clustering result is shown for two clusters (before splitting) with a silhouette score of $0.72$. In \cref{fig:CounterEx_silhScore_after} the clustering result for three clusters (after split at $2.5$ in feature $2$) is shown with a silhouette score of $0.7$. So the clustering quality decreases according to the silhouette score after the split.}
\label{fig:CounterEx_silhScore}
\end{figure}

The change in the silhouette score is the average change in the intra- and inter-distance of each data point. Thus, for the example given, the change in the silhouette score should be influenced by the distances between the three cluster centres. We denote the distance between two cluster centres of clusters $C$ and $S_0$ as $d_{C,S_0}= d_{S_0,C}$. 
\begin{figure}[!t]
    \captionsetup{format=plain, indention=0cm}
    \centering
    \includegraphics[width=0.8\linewidth]{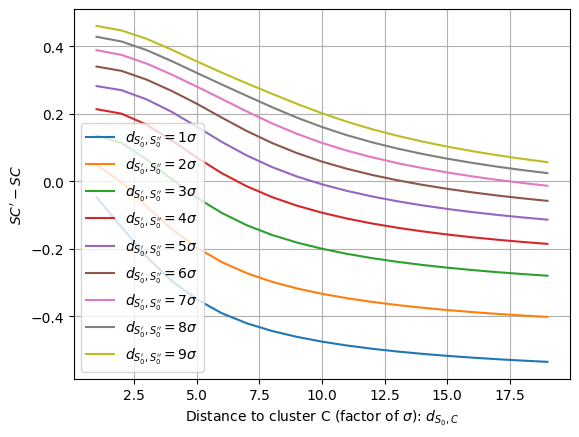}
    \caption{This plot shows the change in the silhouette score of a clustering of a dataset before and after the set $S_0$ is split into $S_0',S_0''$. The data sets vary in the distance between the cluster centres by a factor of the standard deviation of the cluster to understand this effect. Higher values indicate more improvement in the silhouette score and negative values indicate that the silhouette score is lower than before the split. An example setting is given in \cref{fig:CounterEx_silhScore}.}
    \label{fig:overview_counterexample}
\end{figure}
In \cref{fig:overview_counterexample} we analyse the example in \cref{fig:CounterEx_silhScore} for a more general setting regarding the distances of the cluster centres. On the $y$-axis the change in the silhouette score is plotted and on the $x$-axis the distance between the cluster $C$ and $S_0$ (before splitting). The different lines correspond to the distances of $S_0'$ and $S_0''$ (after splitting). The further away the cluster centres of $S_0'$ and $S_0''$ are, the better the change in the silhouette score. For larger $d_{S_0',S_0''}$, the intra-distance before the split increases as well as the inter-distance after the split. As the distance between cluster $C$ and $S_0$ increases, the change in the silhouette score decreases. A larger $d_{C,S_0}$ implies a large inter-distance for the data set before the split as $C$ is the closest cluster. After the split, if $d_{C,S_0'}>d_{S_0',S_0''}$ (analogue for $S_0''$), the closest cluster centre is $S_0''$. Although the intra-distance for $S_0'$ and $S_0''$ decreases after the split, the large inter-distances before the split can not be overcome. Thus, in these cases, the advantage of the split is not captured by the silhouette score. \\

The impact of a DPM-based split on the silhouette value of data points was examined. The conditions under which the silhouette value of a data point improves, depending on whether the data point is in the partitioned subset or in any other partitions, were determined. The limitations of the silhouette score as a clustering metric were discussed using an example in which, even for an optimal split, the overall silhouette score decreased.

%% file: main/xiRho.tex
\section{From \texorpdfstring{$(\xi, \rho)$}{(xi,rho)} to Hero}
\label{cha:xiRho}

The strategy of identifying gaps in data points rather than focusing on the centre of dense areas has yielded encouraging results in terms of balancing the privacy-utility trade-off. As an illustration, DPM employs metrics to locate gaps in data sets. DPM characterises gaps as regions containing a limited number of data points. To enhance the precision of the results, DPM incorporates the metric centreness, which prioritises central gaps over those situated near the dimension bounds.  
In order to gain insight into the potential of the approach of finding gaps, we propose a more theoretical definition of gaps which we call $(\xi, \rho)$-separability.

\subsection{\texorpdfstring{$(\xi, \rho)$}{(xi,rho)}-Separability}
The most intuitive definition of a gap is an area of size $\rho$ in the data points with no data points. If there is a gap in the data set, it can be said that the data set is separable.
% Formal definition of rho-separability
The formalisation of the notion of separability of a set $S$ into two subsets $X$ and $Y$ is based on the requirement that, for every point $x \in X$ there is a zone (an open ball) of radius $\rho/2$ in which no points from $Y$ are present. Similarly, every point $y\in Y$ also has such a zone of radius $\rho/2$ in which no points from $X$ are present.

\begin{definition}[(High dimensional) $\rho$-separability]
\label{def:rhoSep}
    Let $S \subseteq D \in \datasets$ be a set and $x,y \in S$ be two data points. Let $X,Y$ be a partitioning of $S$, i.e. $X \dot\cup Y = S$.
    Then \emph{$S$ is $\rho$-separable into $X$ and $Y$} if $\rho > 0$ is the largest real value such that for every $x \in X$ the open ball $B_{\rho/2}(x):=\{x'\mid\lVert x-x'\rVert_2\leq \rho/2\}$ does not contain elements of $Y$, $ \left |B_{\rho/2}(x)\cap Y \right | = 0$.
    We abbreviate that with \emph{$X,Y$ are $\rho$-separable} and write $\rho_{X,Y}$. If $X,Y$ are clear from the context, we write that $X,Y$ are $\rho$-separable.
\end{definition}

It could be argued that requiring that a gap does not contain any data points is overly restrictive; therefore, we propose an extension to the definition of a gap, whereby the number of data points in the area is also taken into account. A new parameter, denoted as $\xi$, is introduced as an upper bound for the number of data points that can be accepted in an area of size $\rho$. In other words $\xi$ the number of data points that violate $\rho$-separability, and thus these areas can still be considered as gaps in the data set. The set can therefore be defined as $(\xi,\rho)$-separable.
\begin{definition}[$(\rho,\xi)$-separability]
\label{def:rhoxiQSep}
    Given two sets $X,Y \subseteq D \in \datasets$, we say that for $\rho > 0, \xi \in \mathbb{N}_+$, $X$ and $Y$ are \emph{$(\xi,\rho)$-separable} iff. $\xi$ is the smallest natural number such that for every $x \in X$ the open ball $B_{\rho/2}(x)$ does not contain more than $\xi$ elements of $Y$, $ \left |B_{\rho/2}(x)\cap Y \right | \le \xi$.
\end{definition}

\begin{remark}
    If $X$ and $Y$ are $(\xi,\rho)$-separable, then there is at least one point $s \in \R^m$ between $X$ and $Y$ such that the open balls $B_{\rho/2}(s)$ around $s$ have at most $\xi$ points from $X \cup Y$: $|B_{\rho/2}(s) \cap (X\cup Y)| \le \xi$. We say that $s$ is a $(\rho,\xi)$-separator of $X$ and $Y$.
\end{remark}

\subsection{Using Gaps for Clustering}
If a set is $(\xi,\rho)$-separable for reasonable of $\xi$ and $\rho$, it can be assumed to be clusterable. The corresponding gaps can be employed for clustering purposes. One approach that is employed by DPM, for example, is to consider the data points in each dimension separately by projecting all data points onto a unit vector (the axis). In order to link the guarantees for $(\xi, \rho)$-separability to these approaches, we present how they can be interpreted.
We begin by introducing the concept of projecting data points onto a vector. 
\begin{definition}[Projection]
    \label{def:proj}
    Given $v \in \R^m$ and $D \in \datasets$, the \emph{projection} $\pi_v$ of any point $x\in D$ onto $v$ is described as the dot product between $x$ and $v$. Formally, $\pi_v:\R^m \to \R$ with $\pi_v(x)=v\cdot x$. 
    The projection of a set $S\subseteq D$ onto $v$ is described as $\pi_v(S) = \{\pi_v(x)|~ x\in S\}$.
    The \emph{preimage set} on a set $G \subseteq \R$ regarding a vector $v$ is defined as $\pi^{-1}_v(G) = \{x \in \R^m|~ x\cdot v \in G \} \supseteq S$.
\end{definition}

\begin{remark}
   In order to define the vector on which the data points are projected, it is sufficient to consider only the angle of the vector, rather than its length. Consequently, we will only consider $v\in\R^m$ with $||v||_2 = 1$.
\end{remark}

If a data set is $\rho$-separable in one direction $v$ at an interval $G$, then the inverse image of $G$ is empty.
\begin{lemma}
\label{lem:rho-empty}
    Given a data set $D\subseteq \R^n$. For $v\in \R^n$, if $v^TD:=\pi_v(D)$ is $\rho$-separable at the interval $G\subseteq \R$, then $\pi^{-1}_v(G)$ is empty: $D\cap \pi^{-1}_v(G)= \emptyset$.
\end{lemma}

\begin{proof}
    We prove \cref{lem:rho-empty} by contradiction. We thus assume that $G\subseteq \pi_v(\R^n)$ and that $G\cap \pi_v(D)= \emptyset$. We now further assume that there is a data point in $D$ ($x\in \R^n$) that projected to $v$ is also in $G$. Because $v^Tx\in G$, $x$ is also in the inverted image of $G$: $x \in \pi^{-1}_v(G)$. We assumed that $x\in D$ and thereby the projection of $x$ is in the projection of the dataset $D$: $\Rightarrow v^Tx = \pi_v(x)\in \pi_v(D)$. This is a contradiction to the assumption that $ G \cap \pi_v(D) = \emptyset$.
\end{proof}

Let us consider a data set $D$ and a direction $v$. We assume that there is an interval $G$ where no data points are in the inverse image of $G$. In this case, the data set $D$ can be partitioned into two disjoint sets that are $\rho$-separable. This parameter is given by the interval $(a,b)$ of $G$, where $\rho$ is defined as follows $\rho =\frac{b-a}{2}$.
\begin{lemma}
\label{lem:empty-rho}
Let $G=(a,b)$ with $a<b$. Given a data set $D\subseteq \R^n$. For $v\in\R^n$, if $D\cap \pi_v^{-1}(G)= \emptyset$, then there is a partitioning $S\dot \cup S'=D$ of $D$ such that $\forall x \in S, y\in S'$ are $\frac{b-a}{2}$-separable at $\alpha x + (1-\alpha)y$ for some $\alpha \in [0,1]$ and $\alpha x + (1-\alpha)y\in \pi_v^{-1}(G)$.
\end{lemma}
It is assumed that there is an interval $G$ in direction $v$ with a disjoint inverted image to the data set $D$. In accordance with the results of \cref{lem:rho-empty}, the distance between two projected datapoints $\pi_v(x), \pi_v(y)$, which are in two disjoint sets, is required to be at least $b-a$. This result is used to bound the distance between the data points $x$ and $y$ (instead of their projections).
\begin{proof}
Let $S = D \cap \pi^{-1}_v([b,\infty))$ and $S'=D\cap\pi^{-1}_v((-\infty,a])$. $\forall x \in S, y\in S'$ with \cref{lem:rho-empty} implies $\{x,y\}\cap \pi_v^{-1}(G)= \emptyset$. We know that
\begin{align*}
     |\pi_v(x)-\pi_v(y)|&\geq b-a \\
    &= |v^Tx-v^Ty|\geq b-a\\
    &= |v^T(x-y)|\geq b-a \text{.}
\intertext{We can use this to bound $||x-y||^2$ as follows.}
    ||x-y||^2 &= (x-y)^T(x-y)\\
    &=((vv^T)(x-y))^T(vv^T(x-y))\\
    &=(v(v^T(x-y)))^T(v(v^T(x-y)))
\intertext{We need to consider different cases for the relation of $v^T(x-y)$ and $b-a$:}
\intertext{\quad Case 1: $v^T(x-y)\geq b-a$}
    (v(v^T(x-y)))^T(v(v^T(x-y))) &\geq (v(b-a))^T(v(b-a)) \\
    &= (b-a)(b-a)v^Tv \\
    &= (b-a)^2
\intertext{\quad Case 2: $-(v^T(x-y)) \geq b-a$}
    (v(v^T(x-y)))^T(v(v^T(x-y)))&=(v(-(v^T(x-y))))^T(-(v(v^T(x-y)))) \\
    &\geq (v(b-a))^T(v(b-a)) \\
    & = (b-a)^2
\end{align*}
\end{proof}

The argument put forth is that the separation of data points in the centre of a gap that preserves $\xi$ and $\rho$ when projected yields a separation of similar quality in the original dimensionality. In other words, the quality of the separation is preserved when the data points are projected. 
To be more precise, the proof is analogous to that of $\rho$-separability, whereby it is demonstrated that if a one-dimensional separator is constructed using the aforementioned parameters, then there is a $(\xi, \rho)$-separator for the original dimensionality. 

Formally, with \cref{def:proj} we show that if the projection of two sets $X$ and $Y$, is one-dimensionally $(\xi, \rho)$-separable, then already $X$ and $Y$ are $(\xi, \rho)$-separable.
We demonstrate that (1) the value of $\xi$ is maintained even after projection, and (2) a separating $\rho/2$-ball exists within the original set.  

\paragraph{(1) $\xi$ data points are in the preimage of the projection.}
Let $G \subseteq \R$ in $\pi_v(S)$ be a set that contains exactly $\xi$ elements of $S\subseteq D \in \datasets$. The projection onto $v$ is defined as projecting every point $x \in S$ onto $v$. 
With \Cref{def:rhoxiQSep}, we can prove that if $|\pi_v(S)\cap G| = \xi$, then $S$ can be partitioned into some $S_l\dot\cup S_r=S$ s.t. $S_l$ and $S_r$ are $(\rho, \xi)$-separable. If a data set is $(\rho, \xi)$-separable in one direction $v$ in $G$, then the preimage set $\pi^{-1}_v$ of $G$ contains exactly $\xi$ elements.
\begin{lemma}
\label{lem:rhoxi-empty}
    Given a set $S \subseteq D\in \datasets$. For $v\in \R^m$ with $||v||_2 = 1$, if $v^TS:=\pi_v(S)$ is $(\rho, \xi)$-separable in $G\in \R$, then $\pi^{-1}_v(G)$ contains exactly $\xi$ elements: $|S\cap \pi^{-1}_v(G)| = \xi$.
\end{lemma}
\begin{proof} 
    We prove \Cref{lem:rhoxi-empty} by contradiction. Therefore, we assume that $G\in \pi_v(\R^n)$, $G\cap \pi_v(S)= E$ and $|E|=\xi$. Further we assume that there is a data point $x$ in $S$ ($x\in \R^m$) that projected on $v$ is also in $G$ but not in $E$. Because $v^Tx\in G$, $x$ is also in the preimage set of $G$: $x \in \pi^{-1}_v(G)$. 
    As $x\in S$ and thereby the projection of $x$ is in the projection of the dataset $S$: $\Rightarrow v^Tx = \pi_v(x)\in \pi_v(S)$. But with $x \not \in E$ it follows that $G \cap \pi_v(S) = E \cup x$ which is a contradiction to the assumption that $ G \cap \pi_v(S) = E$.
\end{proof}
 
\paragraph{(2) There is a $\rho/2$-open ball in the preimage of the projection.} Assume for a set $S$ and a direction $v$ there is an interval $G$ with $\xi$ data points in the preimage set of $G$. Then, the set $S$ can be partitioned into two disjoint sets that are $(\rho, \xi)$-separable and $\rho = |G| = \max(G) - \min(G)$.

\begin{lemma}
\label{lem:empty-rhoxi}
Let $G \subset \R$  and $S \subseteq D\in \datasets$. For $v\in\R^m$, if $E := S\cap \pi_v^{-1}(G)$, then there is a partitioning $S_l\dot \cup S_r=S$ such that for all pairs $(x,y) $ with $x \in S_l, y\in S_r$, $(x,y)$ are $(\rho, \xi)$-separable with $\rho=|G|, \xi = |E|$. 
\end{lemma}

\begin{proof}
Let $S_r = (S\setminus E)\cap \pi^{-1}_v([\max(G),\infty))$ and $S_l= (S \setminus E) \cap\pi^{-1}_v((-\infty,\min(G)])$. With \Cref{lem:rho-empty} it follows that $\forall x \in S_r, y\in S_l$ the following holds $\{x,y\}\cap \pi_v^{-1}(G)\in \emptyset$. We know that

\begin{align*}
     |\pi_v(x)-\pi_v(y)|\geq& \max(G)-\min(G) \\
    =& |v^Tx-v^Ty|\geq \max(G)-\min(G)\\
    =& |v^T(x-y)|\geq \max(G)-\min(G) \text{.}
\intertext{We use this to bound $||x-y||^2$.}
    ||x-y||^2 =& (x-y)^T(x-y)\\
    =&((vv^T)(x-y))^T(vv^T(x-y))\\
    =&(v(v^T(x-y)))^T(v(v^T(x-y)))
\intertext{We need to consider two different cases regarding the relation of $v^T(x-y)$ and $\max(G)-\min(G) $:}
\intertext{\quad Case 1: $v^T(x-y)\geq \max(G)-\min(G) $}
    (v(v^T(x-y)))^T(v(v^T(x-y)))\geq &(v(\max(G)-\min(G)))^T(v(\max(G)-\min(G))) \\
     =& (\max(G)-\min(G))(\max(G)-\min(G))v^Tv \\
    = & (\max(G)-\min(G))^2
\intertext{\quad Case 2: $-v^T(x-y) \geq \max(G)-\min(G)$}
    (v(v^T(x-y)))^T(v(v^T(x-y)))=&-(v((v^T(x-y))))^T(-(v(v^T(x-y)))) \\
    =&(v(-(v^T(x-y))))^T(-v(v^T(x-y))) \\
    \geq &(v(\max(G)-\min(G)))^T(v(\max(G)-\min(G))) \\
    = & (\max(G)-\min(G))^2
\end{align*}
\end{proof}

Then, with \Cref{lem:empty-rhoxi} the distance between two projected datapoints $\pi_v(x), \pi_v(y)$, which are in two disjoint sets and not in the set $E$ is required to be at least $|G| \ge \lVert \pi_v(x) - \pi_v(y)\rVert_2 $. We can use $|G|$ to give a lower bound for the distance between the data points $x$ and $y$: $\lVert x - y\rVert_2\ge \rho \ge \lVert \pi_v(x) - \pi_v(y)\rVert_2$. 

It has been demonstrated that the methodology of initially projecting the data points onto one of the axes,  selecting a separator, and then projecting them back also yields a separator of comparable quality for the original data points in terms of $\rho, \xi$. 

\subsection{Relation to DPM}

The clustering algorithm DPM \cite{DPM} identifies gaps that can be used to separate clusters. Accordingly, as previously outlined in \cref{cha:introduction}, DPM employs the subscores centreness and emptiness to assess the suitability of areas as potential splits. DPM considers only split candidates of the same size, $\beta$, and thus the size of the split candidates is not taken into account when selecting a split candidate. The subscore centreness criterion prioritises splits situated at the centre of the data points in comparison to those located closer to the border of the range of a given dimension. The subscore emptiness ensures that those splits are favoured that separate groups of data points. The emptiness of a given split is determined by the proportion of data points situated in the vicinity of the split. It should be noted that the emptiness subscore is calculated by subtracting the optimal emptiness of $1$ from the value in question, as the objective is to maximise the scoring function. The parameter $\xi$, represents the number of data points in a gap of size $\rho$. Consequently, for a fixed gap size of $\rho=\beta$,  the emptiness subscore for a split $s$ regarding a set $S$ with noisy count $\tilde n$ can be expressed as follows:
\[e_\beta(S,\tilde n, s) = 1-\frac{|s|}{\tilde n} = 1-\frac{\xi}{\tilde n} \text{ and } \beta=\rho\]

% Conclusion
It can therefore be posited that DPM is a DP clustering mechanism that identifies optimal $(\xi, \rho)$-separations for a fixed value of $\rho$, while accounting for the aspect of centreness.

% Notes on future work: different $\rho$-> problems? ExpMech
A variant of DPM that incorporates $(\xi,\rho)$-separability and centreness for varying values of $\rho$ may yield more optimal  split candidates. Furthermore, the scoring of the splits represents is more accurate with regard to the $(\xi,\rho)$-separability. Nevertheless, although different values of $\rho$ may result in a superior scoring function, a split with lower score may be selected with a higher probability. As the number of candidates increases, the selection process using the Exponential Mechanism becomes increasingly challenging. The probability of a candidate being selected is determined by the ratio of their score to the sum of scores of all split candidates. If the denominator, representing the summed scores of all split candidates, is excessively large, the impact of the nominator (the candidate score) is diminished, leading to a convergence towards a uniform distribution for the selection probability.

%% file: main/conclusions.tex
\section{Conclusion}
\label{cha:conclusions}

% Summary 
The theoretical utility analysis of the differentially private mechanism (DPM) was extended. The authors of DPM have already established the probability of a good split being selected and of DPM halting. In this study, we expanded the analysis of the stopping criterion and introduced the interpretation of the utility of a mechanism through the lens of the silhouette score. Finally, we undertook a comprehensive examination of the underlying concept of DPM, which involves identifying gaps rather than dense areas, with a view to assessing its suitability for clustering. \\

% Analysis of stopping guarantees
In order to provide a more accurate assessment of the stopping guarantees of DPM, we have established a more precise lower bound for the probability that DPM will halt. This is achieved by considering not only the probability that DPM will halt immediately, but also the recursive nature of DPM. Consequently, for each recursion level, we established a lower bound for the probability that DPM will halt at that level, as well as for the probability of such a partitioning occurring. We considered two distinct settings. Initially, we provided a lower bound that is universally applicable, irrespective of the input data set. This is a relatively loose bound. Subsequently, with certain assumptions regarding the input distribution, namely the availability of central splits, we provided more precise lower bounds for the stopping guarantees. In order to gain a more accurate understanding of the guarantees and their implications for the selection of hyperparameters, we conducted an analysis of the inputs that would cause DPM to halt appropriately. We examined the stopping behaviour of DPM for an equally distributed data set as well as a Gaussian distributed data set. Our findings indicated that, for the latter, a large $\alpha$ is necessary to guarantee that DPM halts for reasonable minimum cluster sizes, which in turn affects the utility of DPM.
Our analysis of the stopping behaviour of DPM suggests that a greater number of central splits increases the probability that DPM will halt at a later recursion level. In the current implementation of DPM, the split candidates are only computed for the input data set and not adapted for each subset. This approach saves computational overhead and privacy budget, but also allows for cases where only a few split candidates actually capture the current subset. It would be interesting to implement an adaptation of the split intervals to the current subset and evaluate this approach empirically. \\

% Silhouette Score
Prior to this work, the theoretical analysis of the utility guarantees of DPM are limited to analysing the stopping behaviour and bounding the probability of selecting a split with specific characteristics.  Although this analysis is crucial for comprehending the behaviour of DPM, there is no direct correlation between the selection of an optimal split and the metrics that quantify the utility of a clustering result. As previous research has demonstrated the limitations of the clustering metric inertia, this study instead focuses on the silhouette score.  The point-wise change in the silhouette score after DPM-based splits is analysed. A detailed analysis of all possible settings is conducted to characterise the circumstances under which the silhouette score improves and when it decreases. To provide an interpretation of the overall silhouette score of a clustering, rather than merely point-wise, we present an illustrative example. The presented example demonstrates that, despite the selection of an optimal split by DPM, the silhouette score may still decline. Therefore, the silhouette score also has limitations in terms of capturing the utility of a clustering result, a point that has been previously highlighted by the authors of DPM. Furthermore, they introduced a metric that measures the difference to a non-privacy-preserving baseline, which appears to be free from the same shortcomings of previous metrics. This metric should be subject to further analysis to ascertain its suitability as a clustering metric. It should be noted that they employed a $k$-means optimisation as a baseline that, by design, optimises the metric inertia. Thus, the introduced metric may be susceptible to the same limitations as inertia. \\

% xi, rho
In contrast to existing work, which focuses on dense areas for differentially private clustering, DPM is the first to adopt an approach based on identifying gaps in a data set. Despite the existence of prior work that implements the identification of gaps for the analysis of data sets, there is a lack of theoretical foundation for this approach in the context of clustering. Accordingly, we considered the general approach of finding gaps and linked it to the notion of separability. The term "gap" is defined in terms of both the number of elements within the gap and the width of the gap itself. It was demonstrated that if data points are separable for a given set of parameters, then a gap may be defined by those same parameters. In implementing a specific definition of gaps, DPM considers only intervals of a fixed size. This allows the scoring function of DPM to be interpreted in terms of separability. An interesting avenue for future research would be to analyse the performance when the width of a gap is introduced as a metric, i.e. split interval of multiple sizes. Note that this approach would result in an increasing number of split candidates, which would present a challenge in selecting a suitable split using the Exponential Mechanism.\\

% Conclusion conclusion
In this work, we conducted an exhaustive analysis of the utility guarantees of the differentially private clustering mechanism, DPM. We enhanced the lower bounds for the stopping behaviour of DPM. Furthermore, we established a theoretical connection between DPM and the clustering metric, silhouette score. We demonstrated the potential of the underlying approach of DPM by linking it to the theoretical notion of separability.

%% file: main.bbl
\begin{thebibliography}{10}

\bibitem{balkan}
Maria-Florina Balcan, Travis Dick, Yingyu Liang, Wenlong Mou, and Hongyang Zhang.
\newblock Differentially private clustering in high-dimensional euclidean spaces.
\newblock In {\em International Conference on Machine Learning}, pages 322--331. PMLR, 2017.

\bibitem{lshsplits2021}
Alisa Chang, Badih Ghazi, Ravi Kumar, and Pasin Manurangsi.
\newblock Locally private k-means in one round.
\newblock In {\em ICML}, pages 1441--1451. PMLR, 2021.

\bibitem{DPClEasyIns}
Edith Cohen, Haim Kaplan, Y.~Mansour, Uri Stemmer, and Eliad Tsfadia.
\newblock Differentially-private clustering of easy instances.
\newblock In {\em ICML}, pages 2049--2059. PMLR, 2021.

\bibitem{privacybook}
Cynthia Dwork, Aaron Roth, et~al.
\newblock The algorithmic foundations of differential privacy.
\newblock {\em Foundations and Trends in Theoretical Computer Science}, 9, 2014.

\bibitem{DPClTightApprox}
Badih Ghazi, Ravi Kumar, and Pasin Manurangsi.
\newblock Differentially private clustering: Tight approximation ratios.
\newblock {\em Neurips}, 2020.

\bibitem{Optigrid99}
Alexander Hinneburg and Daniel~A. Keim.
\newblock Optimal grid-clustering: Towards breaking the curse of dimensionality in high-dimensional clustering.
\newblock In {\em Proceedings of the 25th International Conference on Very Large Data Bases}, VLDB '99, page 506–517, San Francisco, CA, USA, 1999. Morgan Kaufmann Publishers Inc.

\bibitem{DPMaxCover}
Matthew Jones, Huy~L. Nguyen, and Thy~D Nguyen.
\newblock Differentially private clustering via maximum coverage.
\newblock {\em Proceedings of the AAAI Conference on Artificial Intelligence}, 35(13), 2021.

\bibitem{mondrian}
Kristen LeFevre, David~J. DeWitt, and Raghu Ramakrishnan.
\newblock Mondrian multidimensional k-anonymity.
\newblock In {\em 22nd ICDE'06}, pages 25--25. IEEE, 2006.

\bibitem{DPM}
Johannes Liebenow, Yara Schütt, Tanya Braun, Marcel Gehrke, Florian Thaeter, and Esfandiar Mohammadi.
\newblock Dpm: Clustering sensitive data through separation, 2024.

\bibitem{emmc}
Huy~L. Nguyen, Anamay Chaturvedi, and Eric~Z Xu.
\newblock Differentially private k-means via exponential mechanism and max cover.
\newblock {\em Proceedings of the AAAI Conference on Artificial Intelligence}, 35(10):9101--9108, 5 2021.

\bibitem{LocDPClKMeans}
Uri Stemmer.
\newblock Locally private k-means clustering.
\newblock {\em Journal of Machine Learning Research}, 22(176):1--30, 2021.

\bibitem{vempala2012randomly}
Santosh~S Vempala.
\newblock Randomly-oriented kd trees adapt to intrinsic dimension.
\newblock In {\em IARCS Annual Conference on Foundations of Software Technology and Theoretical Computer Science (FSTTCS 2012)}. Schloss Dagstuhl-Leibniz-Zentrum fuer Informatik, 2012.

\bibitem{PrivTree}
Jun Zhang, Xiaokui Xiao, and Xing Xie.
\newblock Privtree: {A} differentially private algorithm for hierarchical decompositions.
\newblock {\em CoRR}, 2016.

\end{thebibliography}
